\documentclass[11pt,a4paper]{article}

\usepackage{amsmath,amssymb,amsthm}
\usepackage{amsfonts}
\usepackage{mathtools}
\usepackage{multirow}
\usepackage{placeins}
\usepackage{array}
\usepackage{tikz}
\usepackage[margin=2.5cm]{geometry}
\usepackage[colorlinks=true,linkcolor=blue,citecolor=red,urlcolor=blue]{hyperref}
\allowdisplaybreaks

\newtheorem{theorem}{Theorem}[section]
\newtheorem{corollary}[theorem]{Corollary}
\newtheorem{lemma}[theorem]{Lemma}
\newtheorem{proposition}[theorem]{Proposition}
\theoremstyle{definition}
\newtheorem{definition}[theorem]{Definition}
\newtheorem{remark}[theorem]{Remark}
\newtheorem{example}[theorem]{Example}

\newcommand{\Z}{\mathbb{Z}}
\newcommand{\C}{{\mathcal C}}
\newcommand{\cH}{{\mathcal H}}
\newcommand{\cB}{{\mathcal B}}
\newcommand{\cZ}{{\mathcal Z}}
\newcommand{\cA}{{\mathcal A}}
\newcommand{\cP}{{\mathcal P}}
\newcommand{\zero}{{\mathbf{0}}}
\newcommand{\one}{{\mathbf{1}}}
\newcommand{\uu}{\mathbf{u}}
\newcommand{\vv}{\mathbf{v}}
\newcommand{\ww}{\mathbf{w}}
\newcommand{\xx}{\mathbf{x}}
\newcommand{\yy}{\mathbf{y}}
\newcommand{\zz}{\mathbf{z}}
\newcommand{\cc}{\mathbf{c}}
\newcommand{\ba}{\mathbf{a}}
\newcommand{\wt}{\mathrm{wt}}
\newcommand{\rank}{\operatorname{rank}}
\newcommand{\kernel}{\operatorname{ker}}
\newcommand{\ord}{\operatorname{o}}
\newcommand{\K}{\operatorname{K}}
\newcommand{\hH}{\widehat{H}}
\newcommand{\hcH}{\widehat{\mathcal H}}
\newcommand{\hA}{\widehat{A}}
\newcommand{\hw}{\widehat{\mathbf{w}}}
\newcommand{\hz}{\widehat{\mathbf{z}}}
\newcommand{\hu}{\widehat{\mathbf{u}}}
\newcommand{\hT}{\widehat{T}}
\newcommand{\hcZ}{\widehat{\mathcal Z}}

\begin{document}

\title{\textbf{Full classification of $\mathbb{Z}_{p^s}$-linear generalized\\
Hadamard codes}\thanks{This work has been granted by the Juan de la Cierva 2024 grant (JDC2024-053082-I), funded by Ministerio de Ciencia, Innovación y Universidades (MICIU/AEI/10.13039/501100011033) and co-funded by the European Social Fund Plus (FSE+).}}

\author{
Dipak K. Bhunia\\
\small Departament de Matem\`atiques,\\ \small Universitat Polit\`ecnica de Catalunya,\\
\small Barcelona, Spain
}
\date{}

\maketitle

\begin{abstract}
The $\mathbb{Z}_{p^s}$-additive codes are subgroups of $\mathbb{Z}_{p^s}^n$. A $\mathbb{Z}_{p^s}$-linear generalized Hadamard (GH) code is a GH code over $\mathbb{Z}_p$ which is the Gray map image of a $\mathbb{Z}_{p^s}$-additive code.  A recursive construction gives such codes $H_p^{t_1,\dots,t_s}$, of type $(n;t_1,\dots,t_s)$ and length $p^t$, and a partial classification of them is known. In this paper, we complete the classification of these codes for
every prime $p$.  The main tool is the derived code, obtained by puncturing
the code outside a block of coordinates determined by its kernel.  For a
nonlinear code $H_p^{t_1,\dots,t_s}$ with $t_1\geq2$, the derived code is a
$\mathbb{Z}_{p^{s-1}}$-linear GH code of the same family, although the mechanism for
odd $p$ differs from the binary one.  As a consequence, for $t\geq3$,
there are $\mathcal P(t+1)-\mathcal P(t)$ nonequivalent such codes of length $p^t$ if $p$
is odd, and $\mathcal P(t+1)-\mathcal P(t)-1$ if $p=2$, where $\mathcal P(n)$ is the number of
partitions of $n$.  In particular, this completes the classification of
the $\mathbb{Z}_{2^s}$-linear Hadamard codes.  We also prove a conjecture on the codes
with $t_1\geq2$ and $t_s=0$.
\end{abstract}



\section{Introduction}\label{sec:intro}

Let $\Z_{p^s}$ be the ring of integers modulo $p^s$ with $p$ prime and
$s\geq1$.  The set of $n$-tuples over $\Z_{p^s}$ is denoted by $\Z_{p^s}^n$,
and its elements are also called vectors.  The order of a vector
$\uu\in\Z_{p^s}^n$, denoted by $\ord(\uu)$, is the smallest positive integer
$k$ such that $k\uu=\zero$.  The vectors $\zero$ and $\one$ are the all-zero
and the all-one vectors, the length being always clear from the context.  A
code over $\Z_p$ of length $N$ is a nonempty subset of $\Z_p^N$, and it is
linear if it is a subspace of $\Z_p^N$.  Similarly, a nonempty subset of
$\Z_{p^s}^n$ is a $\Z_{p^s}$-additive code if it is a subgroup of
$\Z_{p^s}^n$.  Note that, when $s=1$, a $\Z_{p^s}$-additive code is a linear
code over $\Z_p$.

The Hamming weight of a vector $\uu\in\Z_p^N$, denoted by $\wt_H(\uu)$, is
the number of its nonzero coordinates.  The Hamming distance of two vectors
$\uu,\vv\in\Z_p^N$ is $d_H(\uu,\vv)=\wt_H(\vv-\uu)$, and the minimum distance
of a code $C$ over $\Z_p$ is $d(C)=\min\{d_H(\uu,\vv):\uu,\vv\in C,\
\uu\neq\vv\}$.  Let $S_N$ be the symmetric group of permutations on
$\{1,\dots,N\}$, and let $\Z_p^*=\Z_p\setminus\{0\}$.  For $\pi\in S_N$ and
$\ba=(a_1,\dots,a_N)\in(\Z_p^*)^N$, we define the monomial map
$\mu_{\pi,\ba}$ of $\Z_p^N$ by $\mu_{\pi,\ba}(\xx)_{\pi(i)}=a_ix_i$ for
$1\leq i\leq N$.  That is, the entry of $\xx$ at the position $i$ is
multiplied by $a_i$ and moved to the position $\pi(i)$.  When $\ba=\one$, we
simply write $\pi$ for $\mu_{\pi,\one}$.  Two codes $C_1$ and $C_2$ over
$\Z_p$ of length $N$ are said to be equivalent if there are a vector
$\yy\in\Z_p^N$ and a monomial map $\mu$ such that
$$
C_2=\{\yy+\mu(\xx):\ \xx\in C_1\}.
$$
They are said to be permutation equivalent if moreover $\yy=\zero$ and $\mu$ is
a permutation.  For $p=2$, every monomial map is a permutation, and we obtain
the usual notion of equivalence of binary codes.  In \cite{HadamardZps}, two
codes are called equivalent when they differ by a monomial map, and in
\cite{EquivZ2s,ZpsEquivalence} when they differ by a translation and a
permutation.  Both notions are particular cases of the one above.  Therefore,
the nonequivalence results of this paper hold for each of them, while all the
equivalences that we establish are permutation equivalences.

Two structural properties of codes over $\Z_p$ are the rank and the dimension
of the kernel.  The rank of a code $C$ over $\Z_p$, denoted by $\rank(C)$, is
the dimension of its linear span $\langle C\rangle$.  Following
\cite{BGH83,pKernel}, the kernel of $C$ is
$$
\K(C)=\{\xx\in\Z_p^N:\ \xx+C=C\},
$$
and its dimension is denoted by $\kernel(C)$.  If $\zero\in C$, then $\K(C)$
is a linear subcode of $C$.  Note also that, if $C$ is linear, then
$\K(C)=C=\langle C\rangle$.  By Lemma \ref{lem:invariance}, two equivalent
codes containing $\zero$ have the same rank and the same dimension of the
kernel.  Hence, these parameters can be used to distinguish between
nonequivalent codes.

A generalized Hadamard (GH) matrix $H(p,\lambda)=(h_{ij})$ of order
$p\lambda$ over $\Z_p$ is a $p\lambda\times p\lambda$ matrix with entries
from $\Z_p$ such that, for every $i$ and $j$ with $1\leq i<j\leq p\lambda$,
the multiset $\{h_{ik}-h_{jk}:1\leq k\leq p\lambda\}$ contains every element
of $\Z_p$ exactly $\lambda$ times \cite{jungnickel1979}.  These matrices are
also known as Butson Hadamard matrices \cite{Butson}, and a GH matrix
$H(2,\lambda)$ over $\Z_2$ corresponds to an ordinary Hadamard matrix of
order $2\lambda$ \cite{Key}.  Two GH matrices are equivalent if one of them
can be obtained from the other by permuting rows and columns and by adding
the same element of $\Z_p$ to all the entries of a row or of a column.
Hence, every GH matrix is equivalent to a normalized one, that is, to one
whose first row and first column contain only zeros.  If $H$ is a normalized
GH matrix and $F_H$ is the set of its rows, then the code
$$
C_H=\bigcup_{\alpha\in\Z_p}(F_H+\alpha\one)
$$
is called a generalized Hadamard (GH) code over $\Z_p$
\cite{dougherty2015ranks}.  A GH code over $\Z_p$ of length $N$ has $pN$
codewords and minimum distance $N(p-1)/p$.  For $p=2$, GH codes are the
usual binary Hadamard codes.

Linear codes over $\Z_{p^s}$ were studied by Blake \cite{Blake} and Shankar
\cite{Shankar}.  Interest in their images under Gray maps increased after it
was established in \cite{Sole} that the Kerdock and Preparata codes can be
represented as Gray images of linear codes over $\Z_4$.  The Gray map from
$\Z_4$ to $\Z_2^2$ given in \cite{Sole} has several generalizations, which
go from $\Z_{2^s}$ to $\Z_2^{2^{s-1}}$
\cite{Carlet,Codes2k,dougherty,Krotov:2007,Nechaev} and, more generally, from
$\Z_{p^s}$ to $\Z_p^{p^{s-1}}$ \cite{GrayIsometry,ShiKrotov2019}.  As in
\cite{KernelZ2s,HadamardZps,ZpsEquivalence}, we use Carlet's Gray map
\cite{Carlet} and its generalization given in \cite{GrayIsometry}, which is
recalled in Section \ref{sec:prelim} and denoted by $\phi_s$.  We write
$\Phi_s:\Z_{p^s}^n\rightarrow\Z_p^{np^{s-1}}$ for its component-wise
extension.  Let $\C$ be a $\Z_{p^s}$-additive code of length $n$.  Its image
$C=\Phi_s(\C)$ is called a $\Z_{p^s}$-linear code, and it has length
$p^{s-1}n$.  As a subgroup of $\Z_{p^s}^n$, the code $\C$ is isomorphic to
$$
\Z_{p^s}^{t_1}\times\Z_{p^{s-1}}^{t_2}\times\cdots\times\Z_p^{t_s}
$$
for some nonnegative integers $t_1,\dots,t_s$, and we say that $\C$, or
equivalently $C$, is of type $(n;t_1,\dots,t_s)$.  Note that
$|\C|=p^{st_1}p^{(s-1)t_2}\cdots p^{t_s}$, and that $\C$ has a generator
matrix with $t_1+\cdots+t_s$ rows, exactly $t_i$ of which have order
$p^{s-i+1}$ for $1\leq i\leq s$.  If $\Phi_s(\C)$ is a GH code, then $\C$ is
called a $\Z_{p^s}$-additive GH code, and $\Phi_s(\C)$ is called a
$\Z_{p^s}$-linear GH code.  For $p=2$, these codes are called
$\Z_{2^s}$-linear Hadamard codes.

A recursive construction of $\Z_{p^s}$-additive GH codes is given in
\cite{KernelZ2s} for $p=2$ and in \cite{HadamardZps} for any prime $p$.
For nonnegative integers $t_1,\dots,t_s$ with $t_1\geq1$, it provides a
$\Z_{p^s}$-additive GH code $\cH_p^{t_1,\dots,t_s}$ of type
$(n;t_1,\dots,t_s)$.  Its image $H_p^{t_1,\dots,t_s}=
\Phi_s(\cH_p^{t_1,\dots,t_s})$ is a $\Z_{p^s}$-linear GH code of length
$p^t$, where
\begin{equation}\label{eq:length}
t+1=\sum_{i=1}^{s}(s-i+1)t_i,
\end{equation}
and we call $(t_1,\dots,t_s)$ the parameters of this code.  For $p=2$, every
$\Z_{2^s}$-linear Hadamard code is equivalent to one of these codes
\cite{Krotov:2007,KernelZ2s}.  As in \cite{HadamardZps,ZpsEquivalence}, by the
$\Z_{p^s}$-linear GH codes of length $p^t$ we mean the codes
$H_p^{t_1,\dots,t_s}$ of that length, with $s\geq1$.  Let $\cA_{t,s,p}$ be the
number of nonequivalent such codes for a fixed $s\geq2$, and let $\cA_{t,p}$ be
the number of nonequivalent such codes when $s\geq2$ is allowed to vary.  Let
also $X_{t,s,p}$ be the number of nonnegative integer solutions
$(t_1,\dots,t_s)$ of \eqref{eq:length} with $t_1\geq1$, and let
$\widetilde X_{t,s,p}$ be the number of those solutions with $t_1\geq2$ for
which $H_p^{t_1,\dots,t_s}$ is nonlinear.

The $\Z_4$-linear Hadamard codes of length $2^t$ are classified by either the
rank or the dimension of the kernel, and there are exactly
$\lfloor(t-1)/2\rfloor$ nonequivalent such codes for any $t\geq3$
\cite{Kro:2001:Z4_Had_Perf,PRV06}.  For $p=2$ and $s\geq3$, the linearity
and the kernel of the codes $H_p^{t_1,\dots,t_s}$ are determined in
\cite{KernelZ2s}.  For odd $p$ and $s\geq2$, they are determined in
\cite{HadamardZps}.  For $p=2$, the rank is computed in \cite{PRV06} for
$s=2$ and in \cite{fernandez2019mathbb} for $s=3$.  In the other cases, the
rank has been computed with \textsc{Magma} \cite{Magma} for small lengths
\cite{KernelZ2s,HadamardZps,ZpsEquivalence}.  In
\cite{KernelZ2s,HadamardZps}, it is also shown that
$\cA_{t,s,2}\leq X_{t,s,2}-1$ for $2\leq s\leq t-1$, and that
$\cA_{t,s,p}\leq X_{t,s,p}$ for odd $p$.  The dimension of the kernel shows
that these bounds are tight for all $t\geq3$ and $2\leq s\leq t-1$ outside
the range given by $t\geq8$ and $3\leq s\leq t-5$.  In that range, the
dimension of the kernel does not classify the codes in general.  It is
conjectured in \cite[Section 7]{HadamardZps} that these bounds are always
tight, and that the rank classifies the codes.  In \cite{EquivZ2s} for
$p=2$, and in \cite{ZpsEquivalence} for any prime $p$, it is shown that some
of these codes with different values of $s$ are permutation equivalent.  This
gives the upper bound
\begin{equation}\label{eq:old-bound}
\cA_{t,p}\leq 1+\sum_{s=2}^{\lfloor(t+1)/2\rfloor}\widetilde X_{t,s,p}
\end{equation}
of \cite{EquivZ2s} and \cite[Theorem 10]{ZpsEquivalence}.  By the
computations of the rank mentioned above, this bound is tight for
$3\leq t\leq11$ if $p=2$, for $3\leq t\leq10$ if $p=3$, and for
$3\leq t\leq8$ if $p=5$.  It is also conjectured in
\cite[Section 3]{ZpsEquivalence} that a code $H_p^{t_1,\dots,t_s}$ with
$t_1\geq2$ and $t_s=0$ is not equivalent to any other code
$H_p^{t'_1,\dots,t'_{s'}}$ of the same length.
Finally, it is pointed out in \cite[Section 5]{ZpsEquivalence} that computing
the rank of the codes with $t_1\geq2$ and $2\leq s\leq\lfloor(t+1)/2\rfloor$
would be enough to obtain the full classification.

Our approach is based on an invariant computed on a part of the code, and it
does not need the rank except in one binary case.  Let $C$ be a code over
$\Z_p$ containing $\zero$.  We say that two coordinates of $C$ lie in the
same kernel block when the vectors of $\K(C)$ which vanish at one of them are
exactly those which vanish at the other one.  For $p=2$, this means that
every vector of $\K(C)$ takes the same value at both coordinates, which is
the definition given in \cite{Z2Z4Z8Profile}.  For odd $p$, we use the vanishing
condition because it is compatible with monomial maps.  Puncturing $C$
outside a set of coordinates means deleting all the coordinates outside this
set.  When the codes obtained by puncturing $C$ outside each of its kernel
blocks are pairwise equivalent, any of them is called the derived code of
$C$.  It is an equivalence invariant of $C$.  When the derived code has the
same property, it can be derived again.  For $p=2$, it is shown in
\cite{OthersProfile} that the kernel blocks of a nonlinear code
$H_2^{t_1,\dots,t_s}$ with $s\geq2$ and $t_1\geq2$ are interchangeable, and
that the code punctured outside any of them is permutation equivalent to
$H_2^{t_1,\dots,t_{s-1}}$.

In this paper, we obtain the full classification of the $\Z_{p^s}$-linear GH
codes for every prime $p$.  We treat all primes together, and we point out
where the binary case and the case of odd primes differ.  We show that the
derived code of a nonlinear code $H_p^{t_1,\dots,t_s}$ with $s\geq2$ and
$t_1\geq2$ is a $\Z_{p^{s-1}}$-linear GH code of the same family.  Since this
descent can be iterated, the classification follows by induction on the length.
However, the mechanism depends on the prime.  For $p=2$, the kernel of such a
code contains a vector which splits each group of $2^{s-1}$ coordinates coming
from one coordinate of the $\Z_{2^s}$-additive code into two halves
\cite{KernelZ2s}, and the derived code is $H_2^{t_1,\dots,t_{s-1}}$.  The
induction then closes in every case except one, in which a $\Z_4$-linear
Hadamard code is compared with a $\Z_8$-linear Hadamard code whose first
parameter is $2$.  In that case, both derived codes are linear, and the known
formulas for the rank decide.  For odd $p$, every vector of the kernel is
constant on each such group \cite{HadamardZps}, and each kernel block is a
union of these groups.  As a consequence, the derived code of
$H_p^{t_1,\dots,t_s}$ is $H_p^{t_1,\dots,t_{s-2},t_{s-1}+1}$.  Compared with
$H_p^{t_1,\dots,t_{s-1}}$, its $\Z_{p^{s-1}}$-additive code has one more
generator of order $p$.  Moreover, for odd $p$, all the codes
$H_p^{t_1,\dots,t_s}$ with $s\geq2$ and $t_1\geq2$ are nonlinear
\cite{HadamardZps}.  Hence, the exceptional case does not occur, and no rank is
needed.  In fact, for odd $p$, the parameters of such a code can be recovered
from the dimensions of the kernels of its successive derived codes.  Figure
\ref{fig:blocks} shows the kernel blocks of a small code for $p=3$ and for
$p=2$.

The main result is Theorem \ref{thm:main}, which states that two codes
$H_p^{t_1,\dots,t_s}$ and $H_p^{t'_1,\dots,t'_{s'}}$ of the same length, with
$s,s'\geq2$ and $t_1,t'_1\geq2$, are equivalent if and only if their
parameters are equal.  We deduce that the bound \eqref{eq:old-bound} is tight
for every prime $p$ and every $t\geq3$.  Moreover, we obtain
$$
\cA_{t,p}=\begin{cases}
\cP(t+1)-\cP(t) & \text{if $p$ is odd},\\
\cP(t+1)-\cP(t)-1 & \text{if } p=2,
\end{cases}
$$
where $\cP(n)$ denotes the number of partitions of $n$, that is, the number
of ways of writing $n$ as a sum of positive integers without regard to order.
In particular, we obtain the full classification of the $\Z_{2^s}$-linear
Hadamard codes.  The parameters $(t_1,\dots,t_s)$ of the codes of length $p^t$
with $t_1\geq2$ correspond to the $\cP(t+1)-\cP(t)$ partitions of $t+1$ with
all parts at least $2$.  The partition of the derived code is obtained from
that of the code by deleting the largest part and, if $p$ is odd, by adding
$1$ to the new largest part.  Figure \ref{fig:young} shows two examples.  Each
equivalence class of the codes of length $p^t$ contains exactly one of these
codes, except that, for $p=2$, the two codes given by $(t+1)$ and $(t-1,2)$
are equivalent.  Figure \ref{fig:classes} shows the classes for $t=6$.  We
also show that $\cA_{t,s,p}=X_{t,s,p}$ if $p$ is odd, and that
$\cA_{t,s,2}=X_{t,s,2}-1$ for $2\leq s\leq t-1$.  This proves the equality
conjectured in \cite[Section 7]{HadamardZps}, although not through the rank.
Finally, we prove the conjecture of \cite[Section 3]{ZpsEquivalence} on the
codes with $t_1\geq2$ and $t_s=0$ for every prime $p$.  For $p=2$, a single
linear code has to be excluded.

The paper is organized as follows.  In Section \ref{sec:prelim}, we recall the
Gray map, the construction of the codes $H_p^{t_1,\dots,t_s}$, and the known
results on their linearity, kernel, rank and equivalences.  In Section
\ref{sec:reduced}, we reduce the classification to the codes with $s\geq2$
and $t_1\geq2$.  In Section \ref{sec:derived}, we introduce the kernel blocks
and the derived code of a code over $\Z_p$.  In Section \ref{sec:blocks}, we
determine the kernel blocks of the codes with $s\geq2$ and $t_1\geq2$, and we
show that they are interchangeable.  In Section \ref{sec:descent}, we
identify the derived code of these codes.  In Section
\ref{sec:classification}, we establish the classification.  In Section
\ref{sec:counting}, we describe the codes and their derived codes in terms of
partitions, and we determine the exact number of nonequivalent codes.
Finally, in Section \ref{sec:conclusions}, we give some conclusions and
further research on this topic.

\section{Preliminaries}\label{sec:prelim}

In this section, we recall the background on the codes
$H_p^{t_1,\dots,t_s}$ that we will use.  From now on, $p$ is a fixed prime.
We write $\cH^{t_1,\dots,t_s}$ and $H^{t_1,\dots,t_s}$ instead of
$\cH_p^{t_1,\dots,t_s}$ and $H_p^{t_1,\dots,t_s}$ when no confusion can
arise.

For $u\in\Z_{p^s}$, we write $[u_0,u_1,\dots,u_{s-1}]_p$ for its $p$-ary
expansion, that is,
$$
u=\sum_{i=0}^{s-1}u_ip^i,
$$
where $u_i\in\{0,\dots,p-1\}$.  We call $u_i$ the digit of index $i$ of
$u$, and we regard it as an element of $\Z_p$ whenever it appears in a
computation modulo $p$.  The Gray map
$\phi_s:\Z_{p^s}\rightarrow\Z_p^{p^{s-1}}$ considered in
\cite{KernelZ2s,HadamardZps,ZpsEquivalence} is
$$
\phi_s(u)=(u_{s-1},u_{s-1},\dots,u_{s-1})+(u_0,\dots,u_{s-2})Y_{s-1},
$$
where $Y_{s-1}$ is a matrix of size $(s-1)\times p^{s-1}$ whose columns are
all the vectors of $\Z_p^{s-1}$ \cite{GrayIsometry}.  For $p=2$, it is
Carlet's Gray map \cite{Carlet}.  We index the $p^{s-1}$ coordinates of
$\phi_s(u)$ by the columns of $Y_{s-1}$, that is, by the vectors
$h=(h_0,\dots,h_{s-2})\in\Z_p^{s-1}$.  With this convention, the definition
of $\phi_s$ becomes
\begin{equation}\label{eq:gray}
\phi_s(u)_h=u_{s-1}+\sum_{i=0}^{s-2}h_iu_i,
\end{equation}
which is the form we will always use.  For $s=1$, the vector $h$ is empty,
and $\phi_1$ is the identity map of $\Z_p$.  As in Section \ref{sec:intro}, we
denote by $\Phi_s$ the component-wise extension of $\phi_s$.

We will use the following property \cite[Lemma 1]{HadamardZps}.  Let
$u\in\Z_{p^s}$ and $\lambda\in\Z_p$.  Then,
\begin{equation}\label{eq:top}
\phi_s\bigl(u+\lambda p^{s-1}\bigr)=\phi_s(u)+\lambda\one.
\end{equation}
In particular, $\phi_s(\lambda p^{s-1})=\lambda\one$.

The next lemma compares $\phi_s$ with $\phi_{s-1}$.  It is the main
computation behind the descent of Section \ref{sec:descent}.  For
$W'\in\Z_{p^{s-1}}$, we compute $pW'\in\Z_{p^s}$ from any integer
representative of $W'$.  This element is well defined, since two such
representatives differ by a multiple of $p^{s-1}$.

\begin{lemma}\label{lem:shift}
Let $s\geq2$, let $\varepsilon\in\{0,\dots,p-1\}$ and $W'\in\Z_{p^{s-1}}$,
and let $W=\varepsilon+pW'\in\Z_{p^s}$.  Then, for every
$h=(h_0,h_1,\dots,h_{s-2})\in\Z_p^{s-1}$, we have
$$
\phi_s(W)_h=\phi_{s-1}(W')_{h'}+\varepsilon h_0, \quad \text{where }\ h'=(h_1,\dots,h_{s-2}).
$$
\end{lemma}

\begin{proof}
If $s=2$, then $h'$ is empty, and $\phi_1(W')_{h'}$ means $\phi_1(W')=W'$.
Let $W'_0,\dots,W'_{s-2}$ be the digits of $W'$.  Since
$0\leq\varepsilon\leq p-1$, the digits of $W$ are $W_0=\varepsilon$ and
$W_\nu=W'_{\nu-1}$ for $1\leq\nu\leq s-1$.  Therefore, by \eqref{eq:gray} for
the alphabets $\Z_{p^s}$ and $\Z_{p^{s-1}}$, we obtain
$$
\phi_s(W)_h=W_{s-1}+h_0W_0+\sum_{\nu=1}^{s-2}h_\nu W_\nu
=W'_{s-2}+\sum_{k=0}^{s-3}h_{k+1}W'_k+\varepsilon h_0
=\phi_{s-1}(W')_{h'}+\varepsilon h_0,
$$
which proves the statement.
\end{proof}

We now recall the construction of \cite{KernelZ2s,HadamardZps}.  For $s\geq1$
and $1\leq i\leq s$, let
\begin{equation}\label{eq:Ti}
T_i=\bigl\{k\cdot p^{i-1}:\ 0\leq k\leq p^{s-i+1}-1\bigr\}=p^{i-1}\Z_{p^s}.
\end{equation}
Note that $T_1=\Z_{p^s}$, and that $T_i$ is the ideal of $\Z_{p^s}$ formed by
the elements of order at most $p^{s-i+1}$.  Let $t_1,\dots,t_s$ be
nonnegative integers with $t_1\geq1$, and put $m=t_1+\cdots+t_s$.  Let
$A^{t_1,\dots,t_s}$ be the matrix with $m$ rows whose columns are exactly,
each of them once, all the vectors $\zz^T$ with
\begin{equation}\label{eq:columns}
\zz\in\{1\}\times T_1^{t_1-1}\times T_2^{t_2}\times\cdots\times T_s^{t_s}.
\end{equation}
The order of the columns only affects the order of the coordinates of the
resulting codes, and we fix one order.  We denote the rows of
$A^{t_1,\dots,t_s}$ by $\ww_1,\dots,\ww_m$.  The level of the row $\ww_i$ is
the index $j_i$ such that
$$
t_1+\cdots+t_{j_i-1}<i\leq t_1+\cdots+t_{j_i}.
$$
Thus, the first $t_1$ rows have level $1$, the next $t_2$ rows have level
$2$, and so on.  Since the first entry of every column is $1$, we have
$\ww_1=\one$, whose order is $p^s$.  For $i\geq2$, the entries of $\ww_i$ run
over $T_{j_i}$, which contains $p^{j_i-1}$.  Hence,
$\ord(\ww_i)=p^{s-j_i+1}$ for $1\leq i\leq m$, and $A^{t_1,\dots,t_s}$ has
exactly $t_j$ rows of order $p^{s-j+1}$ for $1\leq j\leq s$.

Let $\cH^{t_1,\dots,t_s}$ be the $\Z_{p^s}$-additive code generated by the
rows of $A^{t_1,\dots,t_s}$, and let
$H^{t_1,\dots,t_s}=\Phi_s(\cH^{t_1,\dots,t_s})$.  The code
$\cH^{t_1,\dots,t_s}$ is of type $(n;t_1,\dots,t_s)$
\cite{KernelZ2s,HadamardZps}.  Its length is the number of columns, that is,
$$
n=\prod_{i=2}^{m}|T_{j_i}|=\prod_{i=2}^{m}p^{s-j_i+1}.
$$
Since there are $t_j$ rows of level $j$, and since $j_1=1$, we obtain
$n=p^{t+1-s}$, where $t$ is given by \eqref{eq:length}.  Both
$\cH^{t_1,\dots,t_s}$ and $H^{t_1,\dots,t_s}$ contain $\zero$.

\begin{theorem}[\cite{KernelZ2s,HadamardZps}]\label{thm:GH}
Let $s\geq1$, and let $t_1,\dots,t_s$ be nonnegative integers with
$t_1\geq1$.  Then, $H^{t_1,\dots,t_s}$ is a GH code over $\Z_p$ of length
$p^t$, with $t$ given by \eqref{eq:length}.
\end{theorem}

The case $s=1$ is degenerate but useful.  In that case, $T_1=\Z_p$ and
$\phi_1$ is the identity.  Hence, $H^{t_1}$ is the linear code over $\Z_p$ of
length $p^{t_1-1}$ generated by the matrix whose columns are all the vectors
of $\{1\}\times\Z_p^{t_1-1}$.  It is the first order generalized Reed--Muller
code, and it has dimension $t_1$ \cite{Key}.  For $p=2$, it is the first order
Reed--Muller code, and every binary linear Hadamard code of length
$2^{t_1-1}$ is permutation equivalent to it \cite{WMcwill}.

Every codeword of $\cH^{t_1,\dots,t_s}$ is a combination
$\uu=\sum_{i=1}^{m}x_i\ww_i$ with $x_i\in\Z_{p^s}$.  Let $\cZ$ be the set of
columns \eqref{eq:columns}, which we identify with the set of coordinates of
$\cH^{t_1,\dots,t_s}$, and let $\uu(\zz)$ be the entry of $\uu$ at
$\zz\in\cZ$.  Then,
\begin{equation}\label{eq:entry}
\uu(\zz)=\sum_{i=1}^{m}x_iz_i.
\end{equation}
Each $\zz\in\cZ$ produces $p^{s-1}$ coordinates of $H^{t_1,\dots,t_s}$, which
we index by the vectors $h\in\Z_p^{s-1}$ as above.  We write $(\zz,h)$ for
the corresponding coordinate, and we call
$$
F_\zz=\{(\zz,h):\ h\in\Z_p^{s-1}\}
$$
the fiber of $\zz$.  By \eqref{eq:gray}, the codeword $\Phi_s(\uu)$ takes at
$(\zz,h)$ the value
\begin{equation}\label{eq:value}
\Phi_s(\uu)_{(\zz,h)}=\phi_s\bigl(\uu(\zz)\bigr)_h.
\end{equation}

\begin{example}\label{ex:construction}
Let $s=2$ and $(t_1,t_2)=(2,1)$.  Then, $m=3$, the rows $\ww_1$ and $\ww_2$
have level $1$, and the row $\ww_3$ has level $2$.  The columns of $A^{2,1}$
are the $p^3$ vectors $(1,a,pb)^T$ with $a\in\Z_{p^2}$ and
$b\in\{0,\dots,p-1\}$.  Hence, $\ord(\ww_1)=\ord(\ww_2)=p^2$ and
$\ord(\ww_3)=p$, and $H^{2,1}$ has length $p^3\cdot p=p^4$.  This agrees with
\eqref{eq:length}, since $2\cdot2+1=5=t+1$.  The codeword
$\uu=x_1\ww_1+x_2\ww_2+x_3\ww_3$ of $\cH^{2,1}$ takes the value
$W=x_1+x_2a+px_3b$ at the column $\zz=(1,a,pb)$.  By \eqref{eq:value}, the
codeword $\Phi_2(\uu)$ takes the value $W_1+hW_0$ at the coordinate
$(\zz,h)$, for $h\in\Z_p$.  For instance, let $p=3$ and $\uu=\ww_2$.  At the
column $\zz=(1,5,0)$, we have $W=5=[2,1]_3$.  Hence, the three coordinates of
the fiber $F_\zz$ take the values $1$, $0$ and $2$, for $h=0$, $1$ and $2$,
respectively.
\end{example}

The linearity of the codes $H^{t_1,\dots,t_s}$ is the first place where the
binary case and the case of odd primes differ.

\begin{theorem}[\cite{Kro:2001:Z4_Had_Perf,PRV06,KernelZ2s,HadamardZps}]
\label{thm:linearity}
Let $s\geq2$.  If $p$ is odd, then $H^{t_1,\dots,t_s}$ is linear if and only
if its parameters are $(1,0,\dots,0,t_s)$.  If $p=2$ and $s=2$, then
$H^{t_1,t_2}$ is linear if and only if $t_1\leq2$.  If $p=2$ and $s\geq3$,
then $H^{t_1,\dots,t_s}$ is linear if and only if its parameters are
$(1,0,\dots,0,t_s)$ or $(1,0,\dots,0,1,t_s)$.
\end{theorem}

For instance, the code $H^{2,1}$ of Example \ref{ex:construction} is linear
if $p=2$, and nonlinear if $p$ is odd.  For $p=2$ and $s=3$, the linear codes
are $H^{1,0,t_3}$ and $H^{1,1,t_3}$.  The kernel of $H^{t_1,\dots,t_s}$ is
determined in \cite{Kro:2001:Z4_Had_Perf,PRV06} for $p=2$ and $s=2$, in
\cite{KernelZ2s} for $p=2$ and $s\geq3$, and in \cite{HadamardZps} for odd
$p$.  We only need the case $t_1\geq2$.

\begin{theorem}[\cite{Kro:2001:Z4_Had_Perf,PRV06,KernelZ2s,HadamardZps}]
\label{thm:kernel}
Let $s\geq2$, and let $H=H^{t_1,\dots,t_s}$ be a nonlinear code with
$t_1\geq2$.  If $p$ is odd, then the $m$ vectors $\Phi_s(p^{s-j_i}\ww_i)$,
with $1\leq i\leq m$, form a basis of $\K(H)$.  If $p=2$, then a basis of
$\K(H)$ is formed by these $m$ vectors together with
$\Phi_s((2^{s-1}-1)\one)$.  In particular, $\kernel(H)=m$ if $p$ is odd, and
$\kernel(H)=m+1$ if $p=2$.
\end{theorem}

Note that $p^{s-j_i}=\ord(\ww_i)/p$.  For odd $p$, the kernel $\K(H)$ is the
Gray map image of the subcode of $\cH^{t_1,\dots,t_s}$ formed by the vectors
of order at most $p$ \cite[Theorem 4]{HadamardZps}.  For instance, the kernel
of the code $H^{2,1}$ of Example \ref{ex:construction}
has dimension $3$ for odd $p$, and it is spanned by
$\Phi_2(p\ww_1)=\one$, $\Phi_2(p\ww_2)$ and $\Phi_2(\ww_3)$.  For $p=2$, the
additional vector $\Phi_s((2^{s-1}-1)\one)$ is the image of
$(1+2+\cdots+2^{s-2})\one$.  The corresponding vector
$\Phi_s((1+p+\cdots+p^{s-2})\one)$ does not belong to the kernel when $p$ is
odd \cite{HadamardZps}.  This difference is the origin of the different
behavior of the kernel blocks in Section \ref{sec:blocks}.

The rank will only be needed for $p=2$ and $s\in\{2,3\}$.  For $s=2$, the
rank of the $\Z_4$-linear Hadamard code $H_2^{t_1,t_2}$ is computed in
\cite{PRV06}.  It is
\begin{equation}\label{eq:rank4}
\rank\bigl(H_2^{t_1,t_2}\bigr)=2t_1+t_2+\binom{t_1-1}{2},
\end{equation}
where $t_1\geq1$ and $t_2\geq0$.  For $s=3$, the rank of the $\Z_8$-linear
Hadamard code $H_2^{t_1,t_2,t_3}$ is computed in \cite{fernandez2019mathbb}.
It is
\begin{equation}\label{eq:rank8}
\rank\bigl(H_2^{t_1,t_2,t_3}\bigr)
=\frac{t_1^{4}-2t_1^{3}+35t_1^{2}+14t_1}{24}
+\frac{t_2}{2}\bigl(t_1^{2}+t_1+t_2+1\bigr)+t_3+1,
\end{equation}
where $t_1\geq1$ and $t_2,t_3\geq0$.

Codes $H^{t_1,\dots,t_s}$ with different values of $s$ can be permutation
equivalent, as the following result shows.

\begin{theorem}[\cite{EquivZ2s,ZpsEquivalence}]\label{thm:equiv}
Let $s\geq2$ and $t_s\geq1$.  Then, the $\Z_{p^s}$-linear GH code
$H^{t_1,\dots,t_s}$ is permutation equivalent to the $\Z_{p^{s+1}}$-linear GH
code
$$
H^{\,1,\,t_1-1,\,t_2,\dots,t_{s-1},\,t_s-1}.
$$
\end{theorem}

The parameters of the second code in Theorem \ref{thm:equiv} are $1$, then
$t_1-1$, then $t_2,\dots,t_{s-1}$, and finally $t_s-1$.  For $s=2$, the
middle segment is empty, and this code is $H^{1,\,t_1-1,\,t_2-1}$.  In
\cite[Theorem 7]{ZpsEquivalence}, the statement is given more generally, with
$\ell$ steps instead of one, and we only need the case $\ell=1$.

\section{Reduced codes}\label{sec:reduced}

In this section, we show that the classification can be restricted to the
codes $H^{t_1,\dots,t_s}$ with $s\geq2$ and $t_1\geq2$.  Read in the opposite
direction, Theorem \ref{thm:equiv} shows that every $\Z_{p^s}$-linear GH code
with $s\geq3$ and first parameter $1$ is permutation equivalent to a
$\Z_{p^{s-1}}$-linear GH code.  This suggests the following definition.

\begin{definition}\label{def:reduced}
A code $H^{t_1,\dots,t_s}$ with $s\geq2$ is called \emph{reduced} if
$t_1\geq2$.
\end{definition}

By Theorem \ref{thm:linearity}, every reduced code is nonlinear if $p$ is odd.
If $p=2$, a reduced code is linear if and only if $s=2$ and $t_1=2$.  This
explains the exceptional case in Section \ref{sec:classification}.

\begin{lemma}\label{lem:lower}
Let $s\geq3$.  Then, the $\Z_{p^s}$-linear GH code $H^{1,t_2,\dots,t_s}$ is
permutation equivalent to the $\Z_{p^{s-1}}$-linear GH code
$$
H^{\,t_2+1,\,t_3,\dots,t_{s-1},\,t_s+1}.
$$
\end{lemma}

\begin{proof}
Put $u_1=t_2+1$, put $u_i=t_{i+1}$ for $2\leq i\leq s-2$, and put
$u_{s-1}=t_s+1$.  These are nonnegative integers with $u_1\geq1$ and
$u_{s-1}\geq1$.  Since $s-1\geq2$, Theorem \ref{thm:equiv} applies to the
code $H^{u_1,\dots,u_{s-1}}$.  Hence, this code is permutation equivalent to
$$
H^{\,1,\,u_1-1,\,u_2,\dots,u_{s-2},\,u_{s-1}-1}
=H^{\,1,\,t_2,\,t_3,\dots,t_{s-1},\,t_s},
$$
by the definition of the $u_i$.
\end{proof}

The parameters of the second code in Lemma \ref{lem:lower} are $t_2+1$, then
$t_3,\dots,t_{s-1}$, and finally $t_s+1$.  For $s=3$, this code is
$H^{\,t_2+1,\,t_3+1}$.  Lemma \ref{lem:lower} can be applied repeatedly.  To
describe the result, we put
$$
R(1,u_2,\dots,u_r)=(u_2+1,u_3,\dots,u_{r-1},u_r+1)
$$
for parameters $(1,u_2,\dots,u_r)$ with $r\geq3$, the middle segment being
empty for $r=3$.  By Lemma \ref{lem:lower}, the code $H^{1,u_2,\dots,u_r}$ is
permutation equivalent to the code $H^{R(1,u_2,\dots,u_r)}$.  Moreover, the
map $R$ is injective.  Indeed, if $(v_1,\dots,v_{r-1})=R(1,u_2,\dots,u_r)$,
then $u_2=v_1-1$, $u_i=v_{i-1}$ for $3\leq i\leq r-1$, and $u_r=v_{r-1}-1$.

\begin{lemma}\label{lem:R}
Let $s\geq2$, and let $H=H^{t_1,\dots,t_s}$ be a nonlinear code.  Then, there
is a smallest integer $k\geq0$ such that $R^k(t_1,\dots,t_s)$ is defined and
its first entry is at least $2$.  Moreover, $k\leq s-2$, and
$R^k(t_1,\dots,t_s)$ are the parameters of a nonlinear reduced
$\Z_{p^{s-k}}$-linear GH code, which is permutation equivalent to $H$ and
whose last parameter is $t_s+k$.
\end{lemma}

\begin{proof}
If $t_1\geq2$, then $k=0$, and there is nothing to prove.  Assume that
$t_1=1$.  By Theorem \ref{thm:linearity}, the code $H^{1,t_2}$ is linear for
every prime $p$.  Hence, $s\geq3$, and $R(t_1,\dots,t_s)$ is defined.  By
Lemma \ref{lem:lower}, the code $H^{R(t_1,\dots,t_s)}$ is permutation
equivalent to $H$.  In particular, it is nonlinear, since a permutation of
the coordinates maps linear codes to linear codes.  Its last parameter is
$t_s+1$.  If its first parameter equals $1$, we repeat the argument with this
code.  Each step replaces $s$ by $s-1$ and increases the last parameter by
one.  Moreover, the argument shows that every nonlinear code
$H^{u_1,\dots,u_r}$ with $u_1=1$ has $r\geq3$.  Hence, the process stops
after $k\leq s-2$ steps, at a nonlinear code with $s-k\geq2$ parameters whose
first parameter is at least $2$.  This code is reduced, it is permutation
equivalent to $H$, and its last parameter is $t_s+k$.
\end{proof}

For instance, $R(1,0,1,2)=(1,1,3)$ and $R(1,1,3)=(2,4)$.  Hence, for odd $p$,
the nonlinear $\Z_{p^4}$-linear GH code $H^{1,0,1,2}$ is permutation
equivalent to the reduced $\Z_{p^2}$-linear GH code $H^{2,4}$, and $k=2$.
For $p=2$, both codes are linear by Theorem \ref{thm:linearity}.  The linear
codes are treated in Proposition \ref{prop:reduction} below, by using the
next lemma.

\begin{lemma}\label{lem:linear}
Let $t\geq1$.  Then, the $\Z_{p^2}$-linear GH code $H^{1,t-1}$ is permutation
equivalent to the linear code $H^{t+1}$.
\end{lemma}

\begin{proof}
All the rows of $A^{1,t-1}$ except the first one have level $2$.  Hence, the
columns of $A^{1,t-1}$ are the vectors $\zz=(1,pc_2,\dots,pc_t)$ with
$c_2,\dots,c_t\in\{0,\dots,p-1\}$.  Let $\uu=\sum_{i=1}^{t}x_i\ww_i$ with
$x_i\in\Z_{p^2}$, and write $x_1=\varepsilon+pq$ with
$\varepsilon,q\in\{0,\dots,p-1\}$.  Then,
$$
\uu(\zz)=\varepsilon+p\,\Bigl(q+\sum_{i=2}^{t}x_ic_i\Bigr).
$$
Therefore, the digits of $\uu(\zz)$ are $\varepsilon$ and
$q+\sum_{i=2}^{t}x_ic_i$ modulo $p$, and \eqref{eq:gray} gives
$$
\Phi_2(\uu)_{(\zz,h)}=q+\sum_{i=2}^{t}x_ic_i+\varepsilon h,
$$
computed modulo $p$, for every $h\in\Z_p$.  This is the value at the column
$(1,c_2,\dots,c_t,h)$ of the codeword of $\cH^{t+1}$ whose coefficients are
$q,x_2,\dots,x_t,\varepsilon$ modulo $p$.  The map
$(\zz,h)\mapsto(1,c_2,\dots,c_t,h)$ is a bijection from the set of
coordinates of $H^{1,t-1}$ onto the set of columns of $A^{t+1}$, which are
all the vectors of $\{1\}\times\Z_p^{t}$.  Moreover, as the $x_i$ run over
$\Z_{p^2}$, the elements $q$, $x_2,\dots,x_t$ modulo $p$, and $\varepsilon$
run independently over $\Z_p$.  Therefore, this bijection carries
$H^{1,t-1}$ onto $H^{t+1}$.
\end{proof}

\begin{proposition}\label{prop:reduction}
Let $H=H^{t_1,\dots,t_s}$ be a code of length $p^t$ with $s\geq2$.  If $H$ is
nonlinear, then it is permutation equivalent to a nonlinear reduced code of
length $p^t$.  If $H$ is linear, then it is permutation equivalent to
$H^{1,t-1}$ and to the linear code $H^{t+1}$.
\end{proposition}

\begin{proof}
The first assertion follows from Lemma \ref{lem:R}, since permutation
equivalent codes have the same length.  Now, assume that $H$ is linear.  By
Lemma \ref{lem:linear}, it is enough to show that $H$ is permutation equivalent
to $H^{t+1}$.  If $p=2$, then $H$ and $H^{t+1}$ are binary linear Hadamard
codes of length $2^t$, so they are permutation equivalent \cite{WMcwill}.  If
$p$ is odd, then the parameters of $H$ are $(1,0,\dots,0,t_s)$ by Theorem
\ref{thm:linearity}, and $t_s=t+1-s$ by \eqref{eq:length}.  If $s\geq3$, then
$R(1,0,\dots,0,t_s)=(1,0,\dots,0,t_s+1)$ has $s-1$ entries.  Applying Lemma
\ref{lem:lower} $s-2$ times, we obtain that $H$ is permutation equivalent to
$H^{1,t_s+s-2}=H^{1,t-1}$.  This also holds for $s=2$, since then
$H=H^{1,t-1}$.  By Lemma \ref{lem:linear}, the code $H$ is permutation
equivalent to $H^{t+1}$.
\end{proof}

\section{Kernel blocks and the derived code}\label{sec:derived}

In this section, we extend the kernel blocks of \cite{Z2Z4Z8Profile} to codes over
$\Z_p$, and we introduce the derived code.  The definitions and their basic
properties are given for an arbitrary code over $\Z_p$ containing the zero
vector.  In Sections \ref{sec:blocks} and \ref{sec:descent}, we determine the
kernel blocks and the derived code of the nonlinear reduced codes.  For
$p=2$, these results are also given in \cite{OthersProfile}.  We include the
binary case, since the same proofs cover every prime.

For $\xx\in\Z_p^N$ and $B\subseteq\{1,\dots,N\}$, we write $\xx|_B$ for the
vector obtained from $\xx$ by deleting the coordinates whose index lies
outside $B$.  For $D\subseteq\Z_p^N$, we put
$$
D|_B=\{\xx|_B:\ \xx\in D\},
$$
and we call it the code $D$ punctured outside $B$.  The entries of $\xx|_B$
are indexed by the elements of $B$.  Two codes whose coordinates are indexed
by different sets of the same size are compared through a bijection between
these sets.  Another bijection only changes the result by a permutation of
the coordinates.  Hence, equivalence and permutation equivalence make sense
for such codes.  Since deleting coordinates is a linear map, every
$D\subseteq\Z_p^N$ satisfies
\begin{equation}\label{eq:span-puncture}
\langle D\rangle\big|_B=\bigl\langle D|_B\bigr\rangle.
\end{equation}

\begin{definition}\label{def:blocks}
Let $C\subseteq\Z_p^N$ be a code with $\zero\in C$, and put $K=\K(C)$.  For
$1\leq i\leq N$, let
$$
K_i=\{\xx\in K:\ x_i=0\}.
$$
For $1\leq i,i'\leq N$, we write $i\sim_K i'$ if and only if $K_i=K_{i'}$.
The classes of the equivalence relation $\sim_K$ on the set of coordinates
are called the \emph{kernel blocks} of $C$, and the set formed by them is
denoted by $\cB(C)$.
\end{definition}

For $p=2$, if $f,g:K\rightarrow\Z_2$ are linear forms with
$\ker(f)=\ker(g)$, then $f=g$.  Applying this to the linear forms
$\xx\mapsto x_i$ and $\xx\mapsto x_{i'}$, we see that Definition
\ref{def:blocks} coincides with the definition given in \cite{Z2Z4Z8Profile}.
There, two coordinates are in the same kernel block when every vector of $K$
takes the same value at both of them.  For odd $p$, we use the condition
$K_i=K_{i'}$ because it is compatible with monomial maps by Lemma
\ref{lem:invariance}.  When $\one\in K$, the condition $K_i=K_{i'}$ is
equivalent to $x_i=x_{i'}$ for every $\xx\in K$, as the next lemma shows.
Hence, it takes the same form as in the binary case.

\begin{lemma}\label{lem:one}
Let $C\subseteq\Z_p^N$ be a code with $\zero\in C$ and $\one\in\K(C)$.  Then,
two coordinates $i$ and $i'$ lie in the same kernel block of $C$ if and only
if $x_i=x_{i'}$ for every $\xx\in\K(C)$.
\end{lemma}

\begin{proof}
Put $K=\K(C)$, and let $f_i:K\rightarrow\Z_p$ be the linear form defined by
$f_i(\xx)=x_i$.  Then, $K_i=\ker(f_i)$, and $f_i\neq0$ because
$f_i(\one)=1$.  If $f_i=f_{i'}$, then $K_i=K_{i'}$.  Conversely, assume that
$K_i=K_{i'}$, that is, $\ker(f_i)=\ker(f_{i'})$.  Two nonzero linear forms on
$K$ with the same kernel are proportional.  Hence, $f_{i'}=cf_i$ for
some $c\in\Z_p^*$.  Evaluating at $\one$, we obtain $c=1$.  Therefore,
$f_i=f_{i'}$.
\end{proof}

The next two lemmas show why the kernel blocks are useful.  The first one
says that an equivalence carries kernel blocks to kernel blocks, and the
second one says that it carries the punctured codes to equivalent punctured
codes.

\begin{lemma}\label{lem:invariance}
Let $C,D\subseteq\Z_p^N$ be codes with $\zero\in C$ and $\zero\in D$, and
assume that $D=\yy+\mu(C)$ for some $\yy\in\Z_p^N$ and some monomial map
$\mu=\mu_{\pi,\ba}$.  Then, $\langle D\rangle=\mu(\langle C\rangle)$ and
$\K(D)=\mu(\K(C))$.  In particular, $C$ and $D$ have the same rank and the
same dimension of the kernel, and $C$ is linear if and only if $D$ is linear.
Moreover, the map $B\mapsto\pi(B)$ is a bijection from $\cB(C)$ onto
$\cB(D)$.
\end{lemma}

\begin{proof}
Let $\vv\in\Z_p^N$.  Since $D=\yy+\mu(C)$, we have $\vv+D=D$ if and only if
$\vv+\mu(C)=\mu(C)$.  Since $\mu$ is a linear bijection, this holds if and
only if $\mu^{-1}(\vv)+C=C$.  Hence, $\K(D)=\mu(\K(C))$.  Since
$\zero\in D$, there is $\cc_0\in C$ with $\yy=-\mu(\cc_0)$, and then
$D=\mu(C-\cc_0)$, where $C-\cc_0=\{\cc-\cc_0:\cc\in C\}$.  The codes $C$ and
$C-\cc_0$ have the same linear span.  Indeed,
$C-\cc_0\subseteq\langle C\rangle$, and every $\cc\in C$ is the difference of
$\cc-\cc_0$ and $\zero-\cc_0$, which belong to $C-\cc_0$.  Hence,
$\langle D\rangle=\mu(\langle C\rangle)$, and the equalities of the ranks and
of the dimensions of the kernels follow.
Moreover, $C$ is linear if and only if $|C|=|\langle C\rangle|$.  Since
$|D|=|C|$ and $|\langle D\rangle|=|\langle C\rangle|$, this condition holds
for $C$ if and only if it holds for $D$.

For the last assertion, let $\xx\in\K(C)$.  The entry of $\mu(\xx)$ at the
position $\pi(i)$ is $a_ix_i$, which vanishes if and only if $x_i=0$.  Hence,
with the notation of Definition \ref{def:blocks}, we have
$\K(D)_{\pi(i)}=\mu(\K(C)_i)$ for $1\leq i\leq N$.  Since $\mu$ is injective,
$\K(D)_{\pi(i)}=\K(D)_{\pi(i')}$ if and only if $\K(C)_i=\K(C)_{i'}$.  That
is, $i$ and $i'$ lie in the same kernel block of $C$ if and only if $\pi(i)$
and $\pi(i')$ lie in the same kernel block of $D$.
\end{proof}

\begin{lemma}\label{lem:punctured}
Let $C$, $D$, $\yy$ and $\mu=\mu_{\pi,\ba}$ be as in Lemma
\ref{lem:invariance}, and let $B\in\cB(C)$.  Then, $C|_B$ and $D|_{\pi(B)}$
are equivalent codes of length $|B|$, both containing $\zero$.  If
$\yy=\zero$ and $\mu$ is a permutation, then they are permutation equivalent.
\end{lemma}

\begin{proof}
By Lemma \ref{lem:invariance}, the set $\pi(B)$ is a kernel block of $D$.
For $\xx\in\Z_p^N$ and $i\in B$, the entry of $\mu(\xx)$ at the position
$\pi(i)$ is $a_ix_i$.  Hence, $\mu(\xx)|_{\pi(B)}$ is obtained from
$\xx|_B$ by multiplying the entry at each position $i\in B$ by $a_i$ and
moving it to the position $\pi(i)$.  This defines a monomial map $\mu_B$ from
the vectors indexed by $B$ to the vectors indexed by $\pi(B)$, and $\mu_B$ is
a permutation when $\ba=\one$.  Therefore,
$$
D|_{\pi(B)}=\bigl(\yy+\mu(C)\bigr)\big|_{\pi(B)}
=\yy|_{\pi(B)}+\mu_B\bigl(C|_B\bigr),
$$
which proves the statement.  Both punctured codes contain $\zero$, because
$\zero\in C$ and $\zero\in D$.
\end{proof}

These two lemmas do not say that all the kernel blocks of a code give
equivalent punctured codes, and in general they do not.  For instance, the
binary linear code $\{000,110,001,111\}$ has the kernel blocks $\{1,2\}$ and
$\{3\}$, and the codes punctured outside them have different lengths.  When
all the punctured codes are equivalent, the punctured code becomes an
invariant of the code itself.

\begin{definition}\label{def:derived}
Let $C\subseteq\Z_p^N$ be a code with $\zero\in C$.  We say that $C$ is
\emph{block-homogeneous} if the codes $C|_B$, with $B$ running over
$\cB(C)$, are pairwise equivalent.  In that case, we write $\partial C$ for
any one of them, and we call it the \emph{derived code} of $C$.  It is well
defined up to equivalence.
\end{definition}

\begin{corollary}\label{cor:derived-invariant}
Let $C$ and $D$ be equivalent codes over $\Z_p$, both containing $\zero$.  If
$C$ is block-homogeneous, then so is $D$, the two codes have the same number
of kernel blocks, all their kernel blocks have the same size, and $\partial C$
and $\partial D$ are equivalent.
\end{corollary}

\begin{proof}
Write $D=\yy+\mu(C)$ with $\mu=\mu_{\pi,\ba}$.  By Lemma \ref{lem:invariance},
every kernel block of $D$ is of the form $\pi(B)$ for a unique
$B\in\cB(C)$, and $|\pi(B)|=|B|$.  By Lemma \ref{lem:punctured}, the code
$D|_{\pi(B)}$ is equivalent to $C|_B$ for every $B\in\cB(C)$.  Since the
codes $C|_B$ are pairwise equivalent, so are the codes $D|_{\pi(B)}$, and all
of them are equivalent to $\partial C$.  Finally, equivalent codes have the
same length, so all these blocks have the same size.
\end{proof}

\section{The kernel blocks of a reduced code}\label{sec:blocks}

In this section, we fix a reduced code $H=H^{t_1,\dots,t_s}$ of length $p^t$,
which we assume to be nonlinear when $p=2$.  We determine its kernel blocks,
and we show that they are interchangeable, that is, that a group of
permutations of the coordinates which map $H$ onto $H$ acts transitively on
them.  We denote by $\kappa$ the dimension of the kernel of $H$.  By Theorem
\ref{thm:kernel}, we have $\kappa=m$ if $p$ is odd, and $\kappa=m+1$ if
$p=2$.  We keep the indexing $(\zz,h)$ of the coordinates of $H$ introduced
in Section \ref{sec:prelim}.

Recall that $j_i$ denotes the level of the row $\ww_i$.  For
$2\leq i\leq m$, the entry $z_i$ of a column $\zz\in\cZ$ lies in $T_{j_i}$ by
\eqref{eq:columns}.  Hence, by \eqref{eq:Ti}, we may write
$z_i=p^{j_i-1}\zeta_i$ for a unique integer $\zeta_i$ with
$0\leq\zeta_i\leq p^{s-j_i+1}-1$.  We put
$$
\chi_i(\zz)=\zeta_i\bmod p\quad\text{and}\quad
\chi(\zz)=\bigl(\chi_2(\zz),\dots,\chi_m(\zz)\bigr)\in\Z_p^{m-1},
$$
where $2\leq i\leq m$.  The digits of $z_i$ are those of $\zeta_i$, shifted
by $j_i-1$ positions.  Hence, the digits of $z_i$ of index smaller than
$j_i-1$ are zero, and $\chi_i(\zz)$ is the digit of index $j_i-1$ of $z_i$.
In particular, $\chi_i(\zz)=0$ if and only if $z_i\in p^{j_i}\Z_{p^s}$.  The
same definition applies to any vector $\delta\in\Z_{p^s}^m$ whose $i$th entry
lies in $T_{j_i}$ for $2\leq i\leq m$.  We now attach to each coordinate
$(\zz,h)$ of $H$ the label
$$
\rho(\zz,h)=\begin{cases}
\chi(\zz) & \text{if $p$ is odd},\\
\bigl(\chi(\zz),\wt_H(h)\bmod2\bigr) & \text{if } p=2.
\end{cases}
$$
In both cases, $\rho(\zz,h)$ belongs to $\Z_p^{\kappa-1}$.

\begin{proposition}\label{prop:blocks}
Let $H=H^{t_1,\dots,t_s}$ be a reduced code with kernel dimension $\kappa$,
which is nonlinear if $p=2$.  Then, two coordinates $(\zz,h)$ and
$(\zz',g)$ lie in the same kernel block of $H$ if and only if
$\rho(\zz,h)=\rho(\zz',g)$.  Moreover, the map $\rho$ is onto
$\Z_p^{\kappa-1}$.  Consequently, $H$ has exactly $p^{\kappa-1}$ kernel
blocks.  If $p$ is odd, every kernel block is a union of fibers.  If $p=2$,
every fiber meets exactly two kernel blocks.
\end{proposition}

\begin{proof}
We evaluate at a coordinate $(\zz,h)$ the basis of $\K(H)$ given by Theorem
\ref{thm:kernel}.  Let $1\leq i\leq m$.  By \eqref{eq:entry}, the entry of
$p^{s-j_i}\ww_i$ at $\zz$ is $p^{s-j_i}z_i$.  For $i\geq2$, this entry is
$p^{s-1}\zeta_i$, and hence it equals $p^{s-1}\chi_i(\zz)$, since
$\zeta_i\equiv\chi_i(\zz)\pmod p$ and $p^{s-1}\cdot p=0$ in $\Z_{p^s}$.  By
\eqref{eq:top}, we have $\phi_s(p^{s-1}\chi_i(\zz))=\chi_i(\zz)\one$.
Therefore, the basis vector $\Phi_s(p^{s-j_i}\ww_i)$ takes the value
$\chi_i(\zz)$ at $(\zz,h)$, whatever $h$ is.  For $i=1$, we have $j_1=1$ and
$\ww_1=\one$, so the first basis vector is $\Phi_s(p^{s-1}\one)=\one$.  If
$p=2$, the remaining basis vector is $\Phi_s((2^{s-1}-1)\one)$.  The digits
of $2^{s-1}-1\in\Z_{2^s}$ are equal to $1$ at the indices $0,\dots,s-2$, and
to $0$ at the index $s-1$.  Hence, by \eqref{eq:gray}, this vector takes the
value $h_0+h_1+\cdots+h_{s-2}$ at $(\zz,h)$, and this value is
$\wt_H(h)\bmod2$.

In particular, $\one\in\K(H)$, and Lemma \ref{lem:one} applies.  The
condition of Lemma \ref{lem:one} is linear in $\xx$, so it is enough to check
it on a basis of $\K(H)$.  By the previous paragraph, $(\zz,h)$ and
$(\zz',g)$ lie in the same kernel block if and only if
$\rho(\zz,h)=\rho(\zz',g)$.

By \eqref{eq:columns}, the entries $z_2,\dots,z_m$ of $\zz\in\cZ$ range
independently over $T_{j_2},\dots,T_{j_m}$.  As $z_i=p^{j_i-1}\zeta_i$ runs
over $T_{j_i}$, the integer $\zeta_i$ runs over $\{0,\dots,p^{s-j_i+1}-1\}$,
which contains all the residues modulo $p$ because $s-j_i+1\geq1$.  Hence,
the map $\chi$ is onto $\Z_p^{m-1}$.  If $p$ is odd, then
$\rho(\zz,h)=\chi(\zz)$ does not depend on $h$.  Hence, $\rho$ is onto
$\Z_p^{m-1}=\Z_p^{\kappa-1}$, and every kernel block is a union of fibers.  If
$p=2$, then, for a fixed $\zz$, the parity of $\wt_H(h)$ takes both values on
the fiber $F_\zz$, since $s-1\geq1$.  Hence, $\rho$ is onto
$\Z_2^{m}=\Z_2^{\kappa-1}$, and every fiber meets exactly two kernel blocks.
In both cases, the kernel blocks are in bijection with $\Z_p^{\kappa-1}$, and
their number is $p^{\kappa-1}$.
\end{proof}

For $c\in\Z_p^{\kappa-1}$, we denote by
$$
B_c=\{(\zz,h):\ \rho(\zz,h)=c\}
$$
the kernel block of $H$ with \emph{label} $c$.

We now show that the kernel blocks are interchangeable.  For a code
$H^{t_1,\dots,t_s}$, we put
$$
G=\{0\}\times T_{j_2}\times\cdots\times T_{j_m}\subseteq\Z_{p^s}^{m},
$$
where $j_i$ is the level of the row $\ww_i$.  This is a subgroup of
$\Z_{p^s}^m$, and \eqref{eq:columns} identifies the set $\cZ$ of columns with
the coset $(1,0,\dots,0)+G$.  In particular, $\zz+\delta\in\cZ$ for every
$\zz\in\cZ$ and every $\delta\in G$.  Hence, for $\delta\in G$, the map
$\pi_\delta$ defined by
$$
\pi_\delta(\zz,h)=(\zz+\delta,h)
$$
is a permutation of the coordinates of $H^{t_1,\dots,t_s}$.  If $p=2$ and
$s\geq2$, we also consider the permutation $\theta$ defined by
$\theta(\zz,h)=(\zz,h+e_0)$, where $e_0=(1,0,\dots,0)\in\Z_2^{s-1}$.  The
next lemma holds for every code $H^{t_1,\dots,t_s}$, reduced or not.

\begin{lemma}\label{lem:auto}
Let $H=H^{t_1,\dots,t_s}$, let $\uu=\sum_{i=1}^{m}x_i\ww_i\in
\cH^{t_1,\dots,t_s}$, and let $\delta\in G$.  Then,
$\pi_\delta\bigl(\Phi_s(\uu)\bigr)=\Phi_s(\uu-\lambda\one),$
where $\lambda=\sum_{i=2}^{m}x_i\delta_i$.  If $p=2$ and $s\geq2$, then
$\theta\bigl(\Phi_s(\uu)\bigr)=\Phi_s\bigl((1+2^{s-1})\uu\bigr).$
Consequently, $\pi_\delta(H)=H$, and $\theta(H)=H$ if $p=2$ and $s\geq2$.
\end{lemma}

\begin{proof}
By the definition of the action of a permutation, the entry of
$\pi_\delta(\Phi_s(\uu))$ at $(\zz,h)$ is the entry of $\Phi_s(\uu)$ at
$(\zz-\delta,h)$, which is $\phi_s(\uu(\zz-\delta))_h$ by \eqref{eq:value}.
Since $\delta_1=0$, \eqref{eq:entry} gives
$$
\uu(\zz-\delta)=\sum_{i=1}^{m}x_i(z_i-\delta_i)=\uu(\zz)-\lambda.
$$
The element $\lambda\in\Z_{p^s}$ depends on $\uu$ and $\delta$, but not on
$\zz$.  Hence, $\pi_\delta(\Phi_s(\uu))=\Phi_s(\uu-\lambda\one)$.

Now, let $p=2$ and $s\geq2$.  The map $\theta$ is an involution of the set of
coordinates.  Hence, the entry of $\theta(\Phi_s(\uu))$ at $(\zz,h)$ is the
entry of $\Phi_s(\uu)$ at $(\zz,h+e_0)$.  Put $W=\uu(\zz)$.  By
\eqref{eq:gray}, this entry is $\phi_s(W)_{h+e_0}=\phi_s(W)_h+W_0$.  On the
other hand, $2^{s-1}W=2^{s-1}W_0$ in $\Z_{2^s}$, and \eqref{eq:top} gives
$\phi_s((1+2^{s-1})W)=\phi_s(W)+W_0\one$.  Therefore,
$\theta(\Phi_s(\uu))=\Phi_s((1+2^{s-1})\uu)$.

Finally, the vector $\uu-\lambda\one$ belongs to $\cH^{t_1,\dots,t_s}$
because $\one=\ww_1$, and so does $(1+2^{s-1})\uu$, because
$\cH^{t_1,\dots,t_s}$ is a subgroup of $\Z_{2^s}^n$.  Hence,
$\pi_\delta(H)\subseteq H$ and $\theta(H)\subseteq H$.  Equality holds because
these maps are injective.
\end{proof}

\begin{proposition}\label{prop:transitive}
Let $H=H^{t_1,\dots,t_s}$ be a reduced code, which is nonlinear if $p=2$, and
let $\Gamma$ be the group generated by the permutations $\pi_\delta$ with
$\delta\in G$, and by $\theta$ if $p=2$.  Then, every element of $\Gamma$
maps $H$ onto $H$, and $\Gamma$ acts transitively on $\cB(H)$.
Consequently, the codes $H|_B$, with $B$ running over $\cB(H)$, are pairwise
permutation equivalent.
\end{proposition}

\begin{proof}
By Lemma \ref{lem:auto}, the generators of $\Gamma$ map $H$ onto $H$, and
hence so does every element of $\Gamma$.  We now describe the action on the
kernel blocks.  Let $\delta\in G$.  For $2\leq i\leq m$, both $z_i$ and
$\delta_i$ lie in $T_{j_i}$, so their digits of index smaller than $j_i-1$
are zero.  Hence, the digit of index $j_i-1$ of $z_i+\delta_i$ is the sum of
the digits of index $j_i-1$ of $z_i$ and $\delta_i$, computed modulo $p$.
That is, $\chi(\zz+\delta)=\chi(\zz)+\chi(\delta)$.  Since $\pi_\delta$ does
not change $h$, Proposition \ref{prop:blocks} gives
$\pi_\delta(B_c)=B_{c+c_\delta}$ for every label $c$, where
$$
c_\delta=\begin{cases}
\chi(\delta) & \text{if $p$ is odd},\\
(\chi(\delta),0) & \text{if } p=2.
\end{cases}
$$
If $p=2$, the permutation $\theta$ does not change $\zz$, and it changes the
parity of $\wt_H(h)$.  Hence, $\theta(B_c)=B_{c+\omega}$, where
$\omega=(0,\dots,0,1)\in\Z_2^{m}$.

Now, fix $i$ with $2\leq i\leq m$, and let $\delta\in G$ be the vector whose
$i$th entry is $p^{j_i-1}$ and whose other entries are zero.  Then,
$\chi_i(\delta)=1$ and $\chi_{i'}(\delta)=0$ for $i'\neq i$.  The map
$\delta\mapsto\chi(\delta)$ is additive on $G$, and its image contains the
standard basis of $\Z_p^{m-1}$.  Hence, it is onto $\Z_p^{m-1}$.  Let $c$ and
$c'$ be two labels.  If $p$ is odd, there is $\delta\in G$ with
$c_\delta=c'-c$, and then $\pi_\delta(B_c)=B_{c'}$.  If $p=2$, there is
$\delta\in G$ such that $c_\delta$ is equal to $c'-c$ or to $c'-c-\omega$.
If $c_\delta=c'-c$, then $\pi_\delta(B_c)=B_{c'}$.  Otherwise,
$c_\delta=c'-c-\omega$, and then $\theta(\pi_\delta(B_c))=B_{c'}$.  Hence,
$\Gamma$ acts transitively on $\cB(H)$.

Finally, let $B,B'\in\cB(H)$, and let $\gamma\in\Gamma$ be such that
$\gamma(B)=B'$.  By Lemma \ref{lem:punctured}, applied with $C=D=H$,
$\yy=\zero$ and $\mu=\gamma$, the codes $H|_B$ and $H|_{B'}$ are permutation
equivalent.
\end{proof}

\begin{corollary}\label{cor:homogeneous}
Let $H$ be a reduced code of length $p^t$ with kernel dimension $\kappa$,
which is nonlinear if $p=2$.  Then, $H$ is block-homogeneous, and all its
kernel blocks have $p^{\,t-\kappa+1}$ coordinates.
\end{corollary}

\begin{proof}
Permutation equivalent codes are equivalent, so $H$ is block-homogeneous by
Proposition \ref{prop:transitive}.  In particular, all the kernel blocks have
the same size.  There are $p^{\kappa-1}$ of them by Proposition
\ref{prop:blocks}, and they partition the $p^t$ coordinates of $H$.  Hence,
each of them has $p^t/p^{\kappa-1}=p^{\,t-\kappa+1}$ coordinates.
\end{proof}

\section{The descent}\label{sec:descent}

In this section, we identify the derived code of a reduced code, which is
well defined by Corollary \ref{cor:homogeneous}.  We keep the notation of
Section \ref{sec:blocks}.  Thus, $H=H^{t_1,\dots,t_s}$ is a reduced code,
which is nonlinear if $p=2$.  Since all the kernel blocks of $H$ are
interchangeable, it is enough to consider one of them.  The most convenient
one is $B_\zero$, the block which contains the coordinate
$((1,0,\dots,0),\zero)$.  Since $\chi_i(\zz)=0$ means that
$z_i\in p^{j_i}\Z_{p^s}$, we have
$$
B_\zero=\begin{cases}
\bigl\{(\zz,h):\ z_i\in p^{j_i}\Z_{p^s}\ \text{for }2\leq i\leq m\bigr\}
& \text{if $p$ is odd},\\[1mm]
\bigl\{(\zz,h):\ z_i\in 2^{j_i}\Z_{2^s}\ \text{for }2\leq i\leq m,
\ \text{and}\ \wt_H(h)\ \text{is even}\bigr\} & \text{if } p=2.
\end{cases}
$$

Let $m'=m-t_s$ be the number of rows of $A^{t_1,\dots,t_s}$ whose level is at
most $s-1$.  We consider the $\Z_{p^{s-1}}$-linear GH code
$$
\hH=\begin{cases}
H^{t_1,\dots,t_{s-2},\,t_{s-1}+1} & \text{if $p$ is odd},\\
H^{t_1,\dots,t_{s-1}} & \text{if } p=2.
\end{cases}
$$
For $s=2$, the code $\hH$ is the linear code $H^{t_1+1}$ if $p$ is odd, and
the linear code $H^{t_1}$ if $p=2$.  Let $\hA$ be the matrix which defines
$\hH$ as in Section \ref{sec:prelim}, let $\hw_1,\hw_2,\dots$ be its rows,
and let $\hcH$ be the $\Z_{p^{s-1}}$-additive code generated by them.  By the
definition of the levels, the row $\hw_i$ has level $j_i$ for $1\leq i\leq m'$.
If $p$ is odd, the matrix $\hA$ has one more row $\hw_{m'+1}$, whose level is
$s-1$.  If $p=2$, it has exactly $m'$ rows.  We write
$\hT_j=p^{j-1}\Z_{p^{s-1}}$ for the sets \eqref{eq:Ti} of the alphabet
$\Z_{p^{s-1}}$, and $\hcZ$ for the set of columns of $\hA$.  By
\eqref{eq:columns}, we have
$$
\hcZ=\begin{cases}
\{1\}\times\hT_{j_2}\times\cdots\times\hT_{j_{m'}}\times\hT_{s-1}
& \text{if $p$ is odd},\\[1mm]
\{1\}\times\hT_{j_2}\times\cdots\times\hT_{j_{m'}} & \text{if } p=2.
\end{cases}
$$
The next lemma gives the bijection between $B_\zero$ and the set of
coordinates of $\hH$ that we will use.

\begin{lemma}\label{lem:beta}
Let $(\zz,h)\in B_\zero$.  Then, $z_i=p^{j_i}\xi_i$ for a unique integer
$\xi_i$ with $0\leq\xi_i\leq p^{s-j_i}-1$ if $2\leq i\leq m'$, and $z_i=0$ if
$m'<i\leq m$.  Moreover, $h_0=h_1+\cdots+h_{s-2}$ if $p=2$.  Let $\beta$ be
the map defined on $B_\zero$ by $\beta(\zz,h)=(\hz,h')$, where
$h'=(h_1,\dots,h_{s-2})$ and
\begin{equation}\label{eq:beta}
\hz=\begin{cases}
\bigl(1,\ p^{j_2-1}\xi_2,\ \dots,\ p^{j_{m'}-1}\xi_{m'},\
p^{s-2}h_0\bigr) & \text{if $p$ is odd},\\[1mm]
\bigl(1,\ 2^{j_2-1}\xi_2,\ \dots,\ 2^{j_{m'}-1}\xi_{m'}\bigr)
& \text{if } p=2.
\end{cases}
\end{equation}
Then, $\beta$ is a bijection from $B_\zero$ onto the set of coordinates of
$\hH$.
\end{lemma}

\begin{proof}
For $s=2$, the vector $h'$ is empty, and the condition
$h_0=h_1+\cdots+h_{s-2}$ for $p=2$ reads $h_0=0$.  Let $2\leq i\leq m$.  If
$i\leq m'$, then $j_i\leq s-1$.  Recall that $z_i=p^{j_i-1}\zeta_i$ with
$0\leq\zeta_i\leq p^{s-j_i+1}-1$.  Since $z_i\in p^{j_i}\Z_{p^s}$, the
integer $\zeta_i$ is divisible by $p$.  Hence, $z_i=p^{j_i}\xi_i$ with
$\xi_i=\zeta_i/p$, and $0\leq\xi_i\leq p^{s-j_i}-1$.  The integer $\xi_i$ is
unique, and it takes all these values as $(\zz,h)$ runs over $B_\zero$.  If
$i>m'$, then $j_i=s$, and $z_i\in p^{s}\Z_{p^s}=\{0\}$.  For
$2\leq i\leq m'$, the map $\xi_i\mapsto p^{j_i-1}\xi_i$ is a bijection from
$\{0,\dots,p^{s-j_i}-1\}$ onto $\hT_{j_i}$, by \eqref{eq:Ti} for the alphabet
$\Z_{p^{s-1}}$.  On $B_\zero$, the integers $\xi_2,\dots,\xi_{m'}$ and the
vector $h'$ vary independently.  If $p$ is odd, the entry $h_0$ also varies
independently, and $h_0\mapsto p^{s-2}h_0$ is a bijection from $\Z_p$ onto
$\hT_{s-1}$.  If $p=2$, the condition that $\wt_H(h)$ is even means that
$h_0=h_1+\cdots+h_{s-2}$, and $h_0$ is determined by $h'$.  In both cases,
$\beta$ is a bijection from $B_\zero$ onto $\hcZ\times\Z_p^{s-2}$, which is
the set of coordinates of $\hH$.
\end{proof}

For $p=2$, we also need the following lemma, which includes the linear
reduced codes.

\begin{lemma}\label{lem:binary-vector}
Let $p=2$, and let $H=H^{t_1,\dots,t_s}$ be a reduced code.  Then,
$\Phi_s((2^{s-1}-1)\one)\in\K(H)$.
\end{lemma}

\begin{proof}
If $H$ is nonlinear, this follows from Theorem \ref{thm:kernel}.  If $H$ is
linear, then $\K(H)=H$.  Moreover, $(2^{s-1}-1)\one=(2^{s-1}-1)\ww_1$ belongs
to $\cH^{t_1,\dots,t_s}$, and hence $\Phi_s((2^{s-1}-1)\one)\in H$.
\end{proof}

\begin{theorem}\label{thm:descent}
Let $H=H^{t_1,\dots,t_s}$ be a reduced code, which is nonlinear if $p=2$.
Then, the bijection $\beta$ of Lemma \ref{lem:beta} carries $H|_{B_\zero}$
onto $\hH$.
\end{theorem}

\begin{proof}
We compare the values of the codewords.  Let $\uu=\sum_{i=1}^{m}x_i\ww_i$
with $x_i\in\Z_{p^s}$, and let $(\zz,h)\in B_\zero$.  Since $z_1=1$, Lemma
\ref{lem:beta} and \eqref{eq:entry} give
$$
W:=\uu(\zz)=x_1+\sum_{i=2}^{m'}p^{j_i}x_i\xi_i.
$$
Write $x_1=\varepsilon+pq$ with $\varepsilon\in\{0,\dots,p-1\}$ and
$q\in\{0,\dots,p^{s-1}-1\}$.  For $2\leq i\leq m'$, let
$r_i\in\{0,\dots,p^{s-j_i}-1\}$ be the residue of $x_i$ modulo $p^{s-j_i}$.
Since $p^{j_i}p^{s-j_i}=0$ in $\Z_{p^s}$, the element $p^{j_i}x_i\xi_i$ only
depends on $x_i$ modulo $p^{s-j_i}$.  Hence, we may replace $x_i$ by $r_i$,
and we obtain $W=\varepsilon+pW'$, where
$$
W'=q+\sum_{i=2}^{m'}p^{j_i-1}r_i\xi_i\in\Z_{p^{s-1}}.
$$
By Lemma \ref{lem:shift}, we obtain
\begin{equation}\label{eq:core}
\phi_s(W)_h=\phi_{s-1}(W')_{h'}+\varepsilon h_0.
\end{equation}
The two cases now differ in the treatment of the term $\varepsilon h_0$.

Assume first that $p$ is odd.  By \eqref{eq:top} for the alphabet
$\Z_{p^{s-1}}$, the right-hand side of \eqref{eq:core} is equal to
$\phi_{s-1}(W'+p^{s-2}\varepsilon h_0)_{h'}$.  Consider the codeword
$$
\hu=q\,\hw_1+\sum_{i=2}^{m'}r_i\hw_i+\varepsilon\,\hw_{m'+1}
$$
of $\hcH$.  By \eqref{eq:entry} for $\hcH$, its entry at the column $\hz$ of
\eqref{eq:beta} is
$$
\hu(\hz)=q+\sum_{i=2}^{m'}r_ip^{j_i-1}\xi_i+\varepsilon p^{s-2}h_0
=W'+p^{s-2}\varepsilon h_0.
$$
Hence, by \eqref{eq:value}, the value of $\Phi_s(\uu)$ at $(\zz,h)$ is equal
to the value of $\Phi_{s-1}(\hu)$ at $\beta(\zz,h)$, for every
$(\zz,h)\in B_\zero$.  That is, $\beta$ carries $\Phi_s(\uu)|_{B_\zero}$ onto
$\Phi_{s-1}(\hu)$.  It remains to show that $\hu$ runs over all of $\hcH$ as
$\uu$ runs over $\cH^{t_1,\dots,t_s}$.  The map $x_1\mapsto(\varepsilon,q)$ is
a bijection from $\Z_{p^s}$ onto $\{0,\dots,p-1\}\times\{0,\dots,p^{s-1}-1\}$,
and the map $x_i\mapsto r_i$ is onto $\{0,\dots,p^{s-j_i}-1\}$ for
$2\leq i\leq m'$.  On the other hand, the row $\hw_1$ has order $p^{s-1}$, the
row $\hw_i$ has order $p^{s-j_i}$ for $2\leq i\leq m'$, and the row
$\hw_{m'+1}$ has order $p$.  Hence, every codeword of $\hcH$ can be written as
$\hu$ for suitable $q$, $r_2,\dots,r_{m'}$ and $\varepsilon$.  Therefore,
$\beta$ carries $H|_{B_\zero}$ onto $\Phi_{s-1}(\hcH)=\hH$.

Assume now that $p=2$, and consider the codeword
$\hu=q\,\hw_1+\sum_{i=2}^{m'}r_i\hw_i$ of $\hcH$.  By \eqref{eq:entry} for
$\hcH$, its entry at the column $\hz$ of \eqref{eq:beta} is $\hu(\hz)=W'$.
Let $\yy=\Phi_{s-1}((2^{s-2}-1)\one)$.  If $s=2$, then $\yy=\zero$ and
$h_0=0$.  If $s\geq3$, the digits of $2^{s-2}-1\in\Z_{2^{s-1}}$ are equal to
$1$ at the indices $0,\dots,s-3$, and to $0$ at the index $s-2$.  Hence, by
\eqref{eq:gray}, the vector $\yy$ takes at $(\hz,h')$ the value
$h_1+\cdots+h_{s-2}$, which is equal to $h_0$ on $B_\zero$.  In both cases,
\eqref{eq:value} and \eqref{eq:core} show that $\beta$ carries
$\Phi_s(\uu)|_{B_\zero}$ onto $\Phi_{s-1}(\hu)+\varepsilon\yy$.  Moreover,
$\yy\in\K(\hH)$.  Indeed, $\yy=\zero$ if $s=2$.  If $s\geq3$, then
$\hH=H^{t_1,\dots,t_{s-1}}$ is a reduced code, and $\yy\in\K(\hH)$ by Lemma
\ref{lem:binary-vector}.  Therefore, $\Phi_{s-1}(\hu)+\varepsilon\yy\in\hH$,
and $\beta$ carries $H|_{B_\zero}$ into $\hH$.  Conversely, the row $\hw_1$
has order $2^{s-1}$, and the row $\hw_i$ has order $2^{s-j_i}$ for
$2\leq i\leq m'$.  Hence, every codeword of $\hcH$ can be written as $\hu$
for suitable $q$ and $r_2,\dots,r_{m'}$.  Taking $\varepsilon=0$, we see that
every codeword of $\hH$ is the image under $\beta$ of a codeword of
$H|_{B_\zero}$.  Therefore, $\beta$ carries $H|_{B_\zero}$ onto $\hH$.
\end{proof}

Combining Theorem \ref{thm:descent} with Corollary \ref{cor:homogeneous}, we
can now record the descent in the language of Section \ref{sec:derived}.

\begin{corollary}\label{cor:derived}
Let $H=H^{t_1,\dots,t_s}$ be a reduced code of length $p^t$ with kernel
dimension $\kappa$, which is nonlinear if $p=2$.  Then, $H$ is
block-homogeneous, and its derived code is, up to equivalence, the
$\Z_{p^{s-1}}$-linear GH code
\begin{equation}\label{eq:derived}
\partial H^{t_1,\dots,t_s}=\begin{cases}
H^{t_1,\dots,t_{s-2},\,t_{s-1}+1} & \text{if $p$ is odd},\\
H^{t_1,\dots,t_{s-1}} & \text{if } p=2.
\end{cases}
\end{equation}
More precisely, $H|_B$ is permutation equivalent to this code for every
$B\in\cB(H)$.  Moreover, $H$ has $p^{\kappa-1}$ kernel blocks, and
$\partial H$ has length $p^{\,t-\kappa+1}$.
\end{corollary}

\begin{proof}
The code $H$ is block-homogeneous by Corollary \ref{cor:homogeneous}, so its
derived code may be computed on the block $B_\zero$.  By Theorem
\ref{thm:descent}, the code $H|_{B_\zero}$ is permutation equivalent to
$\hH$.  By Proposition \ref{prop:transitive}, so is $H|_B$ for every
$B\in\cB(H)$.  The number of kernel blocks is given by Proposition
\ref{prop:blocks}.  Finally, by \eqref{eq:length}, the code $\hH$ has length
$p^{\,\tilde t}$, where
$$
\tilde t+1=\begin{cases}
\sum_{i=1}^{s-1}(s-i)t_i+1 & \text{if $p$ is odd},\\[1mm]
\sum_{i=1}^{s-1}(s-i)t_i & \text{if } p=2.
\end{cases}
$$
On the other hand, by \eqref{eq:length},
$$
\sum_{i=1}^{s-1}(s-i)t_i=\sum_{i=1}^{s}(s-i)t_i
=\sum_{i=1}^{s}(s-i+1)t_i-\sum_{i=1}^{s}t_i=(t+1)-m,
$$
the term of index $i=s$ being zero.  Since $\kappa=m$ if $p$ is odd, and
$\kappa=m+1$ if $p=2$, we obtain $\tilde t=t-\kappa+1$ in both cases.
\end{proof}

\begin{remark}\label{rem:binary}
The difference between the two cases can be seen on the block $B_\zero$.  In
both cases, the $t_s$ rows of $A^{t_1,\dots,t_s}$ of level $s$ vanish on
$B_\zero$, and they do not contribute to the derived code.  For $p=2$, the
kernel contains $\Phi_s((2^{s-1}-1)\one)$, which splits each fiber into two
halves.  Hence, the permutation $\theta$, which acts on $h$, is needed to
show that the kernel blocks are interchangeable.  On $B_\zero$, the entry
$h_0$ is determined by $h_1,\dots,h_{s-2}$, and the term $\varepsilon h_0$ of
\eqref{eq:core} is absorbed by a vector of the kernel of $\hH$.  For odd $p$,
the fibers are not split, and the translations $\pi_\delta$ are enough.  On
$B_\zero$, the entry $h_0$ is free, and it gives the last entry $p^{s-2}h_0$
of the column $\hz$ in \eqref{eq:beta}, that is, the entry of $\hz$ in the
additional row $\hw_{m'+1}$.  This is why, for odd $p$, the matrix $\hA$ has
one more row of order $p$ than $A^{t_1,\dots,t_{s-1}}$.
\end{remark}

\begin{remark}\label{rem:profile}
By \eqref{eq:span-puncture} and Corollary \ref{cor:derived}, the linear span
of $H$ punctured outside any kernel block has dimension $\rank(\partial H)$.
For $p=2$, this was observed in \cite{OthersProfile}. The kernel-block rank
profile of \cite{Z2Z4Z8Profile} extends naturally to codes over $\Z_p$, as the
multiset of these dimensions.  Hence, for every nonlinear reduced code $H$,
it consists of the single value $\rank(\partial H)$, repeated $p^{\kappa-1}$
times.  By Theorem \ref{thm:GH}, the code on the right-hand side of
\eqref{eq:derived} is a GH code of length $p^{\,t-\kappa+1}$.  Hence,
$\partial H$ has $p^{\,t-\kappa+2}$ codewords, and
$\rank(\partial H)\geq t-\kappa+2$.  Equality holds if and only if
$\partial H$ is linear.  Since the codes $H^{t_1}$ are linear, Theorem
\ref{thm:linearity} shows that this happens if and only if $s=2$ when $p$ is
odd, and if and only if $s=2$, or $s=3$ and $t_1=2$, when $p=2$.
\end{remark}

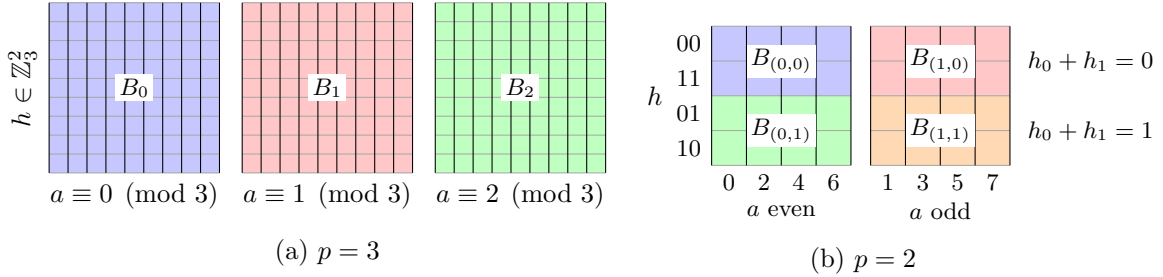
\begin{figure}[!ht]
\centering
\begin{tikzpicture}[font=\small,
  lab/.style={fill=white,inner sep=1.5pt,font=\footnotesize}]
\def\cs{0.25}
\foreach \g/\col/\name in {0/blue!22/B_0,1/red!22/B_1,2/green!25/B_2}{
  \begin{scope}[xshift=\g*2.55cm]
    \fill[\col] (0,0) rectangle (9*\cs,9*\cs);
    \foreach \r in {0,...,9}{\draw[black!35,line width=0.2pt]
      (0,\r*\cs)--(9*\cs,\r*\cs);}
    \foreach \k in {0,...,9}{\draw[line width=0.4pt]
      (\k*\cs,0)--(\k*\cs,9*\cs);}
    \node[lab] at (4.5*\cs,4.5*\cs) {$\name$};
    \node at (4.5*\cs,-0.3) {$a\equiv\g\pmod 3$};
  \end{scope}
}
\node[rotate=90] at (-0.35,4.5*\cs) {$h\in\Z_3^2$};
\node at (3.675,-1.05) {(a) $p=3$};
\begin{scope}[xshift=8.75cm,yshift=0.1cm]
\def\ds{0.46}
\foreach \g/\colA/\colB/\nameA/\nameB in
  {0/blue!22/green!25/{B_{(0,0)}}/{B_{(0,1)}},
   1/red!22/orange!30/{B_{(1,0)}}/{B_{(1,1)}}}{
  \begin{scope}[xshift=\g*2.1cm]
    \fill[\colA] (0,2*\ds) rectangle (4*\ds,4*\ds);
    \fill[\colB] (0,0) rectangle (4*\ds,2*\ds);
    \foreach \r in {0,...,4}{\draw[black!35,line width=0.2pt]
      (0,\r*\ds)--(4*\ds,\r*\ds);}
    \foreach \k in {0,...,4}{\draw[line width=0.4pt]
      (\k*\ds,0)--(\k*\ds,4*\ds);}
    \node[lab] at (2*\ds,3*\ds) {$\nameA$};
    \node[lab] at (2*\ds,1*\ds) {$\nameB$};
  \end{scope}
}
\foreach \k/\a in {0/0,1/2,2/4,3/6,4/1,5/3,6/5,7/7}{
  \node[font=\footnotesize] at ({\k*\ds+0.5*\ds+(\k>3)*(2.1-4*\ds)},-0.22)
  {$\a$};}
\node[font=\footnotesize] at (2*\ds,-0.6) {$a$ even};
\node[font=\footnotesize] at (2.1+2*\ds,-0.6) {$a$ odd};
\foreach \r/\h in {0/10,1/01,2/11,3/00}{
  \node[font=\footnotesize] at (-0.3,\r*\ds+0.5*\ds) {$\h$};}
\node at (-0.75,2*\ds) {$h$};
\node[font=\footnotesize,anchor=west] at (2.1+4*\ds+0.1,3*\ds)
  {$h_0+h_1=0$};
\node[font=\footnotesize,anchor=west] at (2.1+4*\ds+0.1,1*\ds)
  {$h_0+h_1=1$};
\node at (2.0,-1.25) {(b) $p=2$};
\end{scope}
\end{tikzpicture}
\caption{The kernel blocks of $H^{2,0,0}$.  Each column of cells is a fiber
$F_{(1,a)}$, and the cells of the same color form a kernel block.  For $p=2$,
every fiber meets two kernel blocks.}
\label{fig:blocks}
\end{figure}

\begin{example}\label{ex:descent}
We follow the construction on the smallest reduced code with $s=3$.  Let
$(t_1,t_2,t_3)=(2,0,0)$.  Then, $m=2$, both rows have level $1$,
and $t+1=6$ by \eqref{eq:length}.  Hence, $H^{2,0,0}$ has length $p^5$ and
$p^6$ codewords, and it is nonlinear by Theorem \ref{thm:linearity}.  The
columns of $A^{2,0,0}$ are the $p^3$ vectors $(1,a)^T$ with $a\in\Z_{p^3}$.
Hence, the coordinates of $H^{2,0,0}$ are the pairs $((1,a),h)$, where
$h=(h_0,h_1)\in\Z_p^2$.  By Theorem \ref{thm:kernel}, the kernel of
$H^{2,0,0}$ contains $\Phi_3(p^2\ww_1)=\one$ and $\Phi_3(p^2\ww_2)$.  The
second vector takes the value $a\bmod p$ at $((1,a),h)$, since
$p^2a=p^2(a\bmod p)$ in $\Z_{p^3}$.  For a codeword with coefficients
$x_1,x_2\in\Z_{p^3}$, the entry at the column $(1,a)$ is $W=x_1+x_2a$.  On
$B_\zero$, we write $a=p\xi$ with $0\leq\xi\leq p^2-1$.  Let
$x_1=\varepsilon+pq$ as in the proof of Theorem \ref{thm:descent}, and let
$r_2$ be the residue of $x_2$ modulo $p^2$.  Then, $W=\varepsilon+pW'$, where
$W'=q+r_2\xi\in\Z_{p^2}$.

Let $p$ be odd.  Then, $\kappa=2$, and the $p$ kernel blocks are the sets
$$
B_c=\bigl\{((1,a),h):\ a\equiv c\pmod p,\ h\in\Z_p^2\bigr\},
$$
with $c\in\Z_p$, each of them with $p^2\cdot p^2=p^4=p^{\,t-\kappa+1}$
coordinates.  The permutations of Proposition \ref{prop:transitive} are the
maps $\pi_\delta$ with $\delta=(0,d)$ and $d\in\Z_{p^3}$, and $\pi_\delta$
sends ${((1,a),h)}$ to ${((1,a+d),h)}$.  For $d=1$, it maps $B_c$ onto
$B_{c+1}$.  On $B_\zero$, we obtain
$$
\phi_3(W)_h=W'_1+h_1W'_0+\varepsilon h_0
=\phi_2\bigl(W'+p\,\varepsilon h_0\bigr)_{h_1}.
$$
The element $W'+p\,\varepsilon h_0$ is the entry at the column
$(1,\xi,ph_0)$ of the codeword $q\hw_1+r_2\hw_2+\varepsilon\hw_3$ of the
$\Z_{p^2}$-additive code $\cH^{2,1}$.  Therefore, the map
$((1,a),h)\mapsto((1,\xi,ph_0),h_1)$ carries $H^{2,0,0}|_{B_\zero}$ onto
$H^{2,1}$, as Theorem \ref{thm:descent} asserts.  This derived code is the
code of Example \ref{ex:construction}, and it is nonlinear.

Let $p=2$.  Then, $\kappa=3$, and the kernel also contains
$\Phi_3(3\cdot\one)$, which takes the value $h_0+h_1$ at $((1,a),h)$.  Hence,
the four kernel blocks are the sets
$$
B_{(c,e)}=\bigl\{((1,a),h):\ a\equiv c\pmod 2,\ h_0+h_1=e\bigr\},
$$
with $(c,e)\in\Z_2^2$, each of them with $8=2^{\,t-\kappa+1}$ coordinates.
The permutation $\pi_\delta$ with $\delta=(0,1)$ changes $c$, and $\theta$
changes $e$.  On
$B_\zero$, we have $h_0=h_1$, and we obtain
$$
\phi_3(W)_h=W'_1+h_1W'_0+\varepsilon h_1
=\phi_2(W')_{h_1}+\varepsilon\,\phi_2(1)_{h_1}.
$$
The element $W'$ is the entry at the column $(1,\xi)$ of the codeword
$q\hw_1+r_2\hw_2$ of the $\Z_4$-additive code $\cH^{2,0}$, and the vector
$\Phi_2(\one)$ belongs to the linear code $H^{2,0}$.  Therefore, the map
$((1,a),h)\mapsto((1,\xi),h_1)$ carries $H^{2,0,0}|_{B_\zero}$ onto $H^{2,0}$,
which is a binary linear Hadamard code of length $2^3$.

Figure \ref{fig:blocks} shows the kernel blocks of $H^{2,0,0}$ for $p=3$ and
$p=2$.
\end{example}
\section{The classification}\label{sec:classification}

We can now classify the reduced codes.  By Corollary \ref{cor:derived}, the
derived code of a nonlinear reduced $\Z_{p^s}$-linear GH code with $s\geq3$
is a reduced $\Z_{p^{s-1}}$-linear GH code.  If $p$ is odd, this code is
nonlinear, and the descent can be iterated until $s$ becomes $1$.  If
$p=2$, it is linear when $s=3$ and $t_1=2$, and then the descent stops.
This leads to a comparison between a $\Z_4$-linear Hadamard code and a
$\Z_8$-linear Hadamard code whose derived codes are both linear.  We first
deal with this comparison, which only occurs for $p=2$.  It is the only place
where the rank is needed.

\begin{proposition}\label{prop:exceptional}
Let $p=2$, let $H^{a_1,a_2}$ be a nonlinear $\Z_4$-linear Hadamard code, and
let $H^{2,b_2,b_3}$ be a $\Z_8$-linear Hadamard code.  Assume that the two
codes have the same length $2^t$ and the same dimension of the kernel.
Then, $a_1=b_2+4$ and $a_2=b_3-2$.  Moreover,
$$
\rank\bigl(H^{a_1,a_2}\bigr)-\rank\bigl(H^{2,b_2,b_3}\bigr)=a_1-3\geq1.
$$
In particular, the two codes are nonequivalent.
\end{proposition}

\begin{proof}
By Theorem \ref{thm:linearity}, the first code has $a_1\geq3$, and the second
one is nonlinear, since its first parameter is $2$.  Hence, Theorem
\ref{thm:kernel} applies to both of them.  Let $\kappa$ be the common
dimension of the kernel.  By \eqref{eq:length} and Theorem \ref{thm:kernel},
applied to both codes, we have
$$
t+1=2a_1+a_2=6+2b_2+b_3\quad\text{and}\quad\kappa=1+a_1+a_2=3+b_2+b_3.
$$
Hence, $(t+1)-\kappa=a_1-1=3+b_2$, and so $a_1=b_2+4$.  Then,
$a_2=\kappa-1-a_1=(2+b_2+b_3)-(b_2+4)=b_3-2$.  In particular, $a_1\geq4$.

We now compare the ranks.  By \eqref{eq:rank4},
$$
\rank\bigl(H^{a_1,a_2}\bigr)=2a_1+a_2+\frac{a_1^{2}-3a_1+2}{2}.
$$
On the other hand, we put $t_1=2$ in \eqref{eq:rank8}.  Since
$2^4-2\cdot2^3+35\cdot2^2+14\cdot2=168=7\cdot24$, we obtain
$$
\rank\bigl(H^{2,b_2,b_3}\bigr)=7+\frac{b_2}{2}\bigl(2^2+2+b_2+1\bigr)+b_3+1
=8+\frac{b_2(b_2+7)}{2}+b_3.
$$
We substitute $b_2=a_1-4$ and $b_3=a_2+2$, and we get
$$
\rank\bigl(H^{2,b_2,b_3}\bigr)=8+\frac{(a_1-4)(a_1+3)}{2}+a_2+2
=10+a_2+\frac{a_1^{2}-a_1-12}{2}.
$$
Therefore,
$$
\begin{aligned}
\rank\bigl(H^{a_1,a_2}\bigr)-\rank\bigl(H^{2,b_2,b_3}\bigr)
&=2a_1-10+\frac{(a_1^{2}-3a_1+2)-(a_1^{2}-a_1-12)}{2}\\
&=2a_1-10+\frac{14-2a_1}{2}=a_1-3.
\end{aligned}
$$
Since $a_1\geq4$, the two ranks are different.  By Lemma
\ref{lem:invariance}, the two codes are nonequivalent.
\end{proof}

The difference $a_1-3$ would vanish only for $a_1=3$, which is
exactly the case where no such $\Z_8$-linear Hadamard code exists.  We can
now prove the main result of this paper.

\begin{theorem}\label{thm:main}
Let $H=H^{t_1,\dots,t_s}$ and $H'=H^{t'_1,\dots,t'_{s'}}$ be reduced codes of
the same length.  Then, $H$ and $H'$ are equivalent if and only if $s=s'$ and
$t_i=t'_i$ for $1\leq i\leq s$.
\end{theorem}

\begin{proof}
If the parameters are equal, then the two codes are equal.  We prove the
converse by induction on $t$, where $p^t$ is the common length.  Note that a
reduced code of length $p^t$ satisfies $t+1\geq st_1\geq4$ by
\eqref{eq:length}, and hence $t\geq3$.  Let $t\geq3$, and assume that the
converse holds for all reduced codes of smaller length.  Let $H$ and $H'$ be
equivalent reduced codes of length $p^t$.

If $H$ is linear, then $H'$ is linear by Lemma \ref{lem:invariance}.  By
Theorem \ref{thm:linearity}, this is only possible if $p=2$, $s=s'=2$ and
$t_1=t'_1=2$.  Then, \eqref{eq:length} gives $t_2=t'_2=t-3$, and the
parameters are equal.  Hence, we may assume that $H$ and $H'$ are nonlinear.
By Lemma \ref{lem:invariance}, the two codes have the same kernel dimension
$\kappa$.  By Corollary \ref{cor:derived}, both codes are block-homogeneous,
and their derived codes have length $p^{\tilde t}$ with
$\tilde t=t-\kappa+1$.  Moreover, $\partial H$ and $\partial H'$ are
equivalent by Corollary \ref{cor:derived-invariant}.  Since
$\kappa\geq m\geq t_1\geq2$ by Theorem \ref{thm:kernel}, we have
$\tilde t<t$.

Assume first that $s\geq3$ and $s'\geq3$.  By \eqref{eq:derived},
$\partial H$ is a $\Z_{p^{s-1}}$-linear GH code and $\partial H'$ is a
$\Z_{p^{s'-1}}$-linear GH code, with $s-1\geq2$ and $s'-1\geq2$.  Their first
parameters are $t_1\geq2$ and $t'_1\geq2$.  Hence, they are reduced codes of
length $p^{\tilde t}$.  By the induction hypothesis, their parameters are
equal.  In particular, $s-1=s'-1$.  If $p$ is odd, we also obtain
$t_i=t'_i$ for $1\leq i\leq s-2$ and $t_{s-1}+1=t'_{s-1}+1$.  If $p=2$, we
obtain $t_i=t'_i$ for $1\leq i\leq s-1$.  In both cases, $s=s'$ and $t_i=t'_i$
for $1\leq i\leq s-1$.  Then, \eqref{eq:length} gives $t_s=t'_s$, since all the
other terms of $\sum_{i=1}^{s}(s-i+1)t_i$ and $\sum_{i=1}^{s}(s-i+1)t'_i$
coincide.

Assume next that $s=s'=2$.  By \eqref{eq:derived}, the codes $\partial H$ and
$\partial H'$ have lengths $p^{t_1}$ and $p^{t'_1}$ if $p$ is odd, and
$2^{t_1-1}$ and $2^{t'_1-1}$ if $p=2$.  Since both lengths are equal to
$p^{\tilde t}$, we obtain $t_1=t'_1$, and \eqref{eq:length} gives
$t_2=t'_2$.

Finally, assume that exactly one of $s$ and $s'$ equals $2$, say $s=2$ and
$s'\geq3$.  Then, $\partial H$ is linear, and so is $\partial H'$ by Lemma
\ref{lem:invariance}.  On the other hand, $\partial H'$ is a reduced code, as
above.  If $p$ is odd, this is a contradiction, since no reduced code is
linear.  If $p=2$, then Theorem \ref{thm:linearity} gives $s'-1=2$ and
$t'_1=2$.  Hence, $H=H^{t_1,t_2}$ is a nonlinear $\Z_4$-linear Hadamard code,
and $H'=H^{2,t'_2,t'_3}$ is a $\Z_8$-linear Hadamard code of the same length
and with the same dimension of the kernel.  By Proposition
\ref{prop:exceptional}, these two codes are nonequivalent, which is a
contradiction.  Hence, this situation does not occur, and the proof is
complete.
\end{proof}

\begin{example}\label{ex:t7}
For $t\leq6$, the reduced codes of length $p^t$ are $H^{2,0}$ for $t=3$,
$H^{2,1}$ for $t=4$, the codes $H^{2,2}$, $H^{3,0}$ and $H^{2,0,0}$ for
$t=5$, and the codes $H^{2,3}$, $H^{3,1}$ and $H^{2,0,1}$ for $t=6$.  By
Theorems \ref{thm:linearity} and \ref{thm:kernel}, the nonlinear reduced
codes of the same length have different dimensions of the kernel in all
these cases, for every prime $p$.  The first coincidence of kernel dimensions
occurs for $t=7$, between the codes $H^{4,0}$ and $H^{2,0,2}$.  Both codes
have length $p^7$, since $2\cdot4+0=3\cdot2+2\cdot0+2=8$.  Their kernels have
dimension $4$ if $p$ is odd, and $5$ if $p=2$.  By Corollary
\ref{cor:derived}, they also have the same number of kernel blocks, all of
the same size.

For odd $p$, their derived codes are the linear code $H^{5}$ of length $p^4$
and the nonlinear code $H^{2,1}$.  Hence, the two codes are nonequivalent for
every odd prime $p$.  For $p=3$ and $p=5$, this also follows from their
ranks, which are $21$ and $15$ for $p=3$, and $57$ and $31$ for $p=5$
\cite{HadamardZps,ZpsEquivalence}.

For $p=2$, their derived codes are $H^{4}$ and $H^{2,0}$, which are binary
linear Hadamard codes of length $2^3$.  Hence, the descent gives no
information, and this is the smallest instance of Proposition
\ref{prop:exceptional}, with $a_1=4$.  By \eqref{eq:rank4} and
\eqref{eq:rank8}, the ranks of $H^{4,0}$ and $H^{2,0,2}$ are $11$ and $10$,
and the two codes are nonequivalent.
\end{example}

For odd $p$, Theorem \ref{thm:main} also follows from explicit invariants.
For a code $C$ over $\Z_p$ containing $\zero$, we put $\partial^0C=C$ and
$\partial^kC=\partial(\partial^{k-1}C)$ for $k\geq1$, as long as these
derived codes are defined.  For parameters $(t_1,\dots,t_s)$, we also put
$$
\kappa_j=t_1+\cdots+t_j,
$$
where $1\leq j\leq s$.  Note that $\kappa_s=m$.

\begin{proposition}\label{prop:iterated}
Let $p$ be odd, and let $H=H^{t_1,\dots,t_s}$ be a reduced code.  Then, for
$1\leq k\leq s-1$, the code $\partial^kH$ is well defined up to equivalence,
and it is the $\Z_{p^{s-k}}$-linear GH code
\begin{equation}\label{eq:iterated}
\partial^kH=H^{t_1,\dots,t_{s-k-1},\,t_{s-k}+1}.
\end{equation}
For $1\leq k\leq s-2$, this code is reduced, and the dimension of its kernel
is $\kappa_{s-k}+1$.  For $k=s-1$, it is the linear code $H^{t_1+1}$, whose
dimension is $\kappa_1+1$.
\end{proposition}

\begin{proof}
We argue by induction on $k$.  For $k=1$, \eqref{eq:iterated} is Corollary
\ref{cor:derived}.  Let $1\leq k\leq s-2$, and assume that $\partial^kH$ is
well defined and equivalent to $H'=H^{t_1,\dots,t_{s-k-1},t_{s-k}+1}$.  The
code $H'$ has $s-k\geq2$ parameters, and its first parameter is $t_1\geq2$,
since $s-k-1\geq1$.  Hence, $H'$ is reduced.  By Corollary \ref{cor:derived},
it is block-homogeneous, and its derived code is obtained by deleting its
last parameter $t_{s-k}+1$ and adding $1$ to the new last parameter
$t_{s-k-1}$.  That is, $\partial H'=H^{t_1,\dots,t_{s-k-2},t_{s-k-1}+1}$.  By
Corollary \ref{cor:derived-invariant}, every code equivalent to $H'$ is
block-homogeneous, and its derived code is equivalent to $\partial H'$.
Hence, $\partial^{k+1}H$ is well defined, and \eqref{eq:iterated} holds for
$k+1$.

For $1\leq k\leq s-2$, the kernel of the reduced code in \eqref{eq:iterated}
has dimension $t_1+\cdots+t_{s-k}+1=\kappa_{s-k}+1$ by Theorem
\ref{thm:kernel}.  For $k=s-1$, the right-hand side of \eqref{eq:iterated} is
the linear code $H^{t_1+1}$, of dimension $t_1+1=\kappa_1+1$.
\end{proof}

\begin{corollary}\label{cor:odd-invariant}
Let $p$ be odd, and let $H=H^{t_1,\dots,t_s}$ be a reduced code.  Then,
$s-1$ is the smallest integer $k\geq1$ for which $\partial^kH$ is linear, and
the parameters of $H$ are determined by the dimensions of the kernels of
$\partial^kH$ for $0\leq k\leq s-1$.  Moreover, if $H'$ is equivalent to $H$,
then $\partial^kH'$ is well defined and equivalent to $\partial^kH$ for
$0\leq k\leq s-1$.
\end{corollary}

\begin{proof}
By Proposition \ref{prop:iterated}, the code $\partial^kH$ is reduced, and
hence nonlinear, for $1\leq k\leq s-2$, while $\partial^{s-1}H$ is linear.
This gives the first assertion.  By Theorem \ref{thm:kernel} and Proposition
\ref{prop:iterated}, the kernels of $H,\partial H,\dots,\partial^{s-1}H$ have
dimensions $\kappa_s,\kappa_{s-1}+1,\dots,\kappa_1+1$.  Here, we use that the
kernel of a linear code is the whole code.  Since $t_1=\kappa_1$ and
$t_j=\kappa_j-\kappa_{j-1}$ for $2\leq j\leq s$, these dimensions determine
the parameters.  Finally, let $H'$ be equivalent to $H$.  Assume that
$\partial^kH'$ is well defined and equivalent to $\partial^kH$ for some $k$
with $0\leq k\leq s-2$.  The code $\partial^kH$ is equivalent to a reduced
code $C_k$, by hypothesis if $k=0$, and by Proposition \ref{prop:iterated}
otherwise.  By Corollary \ref{cor:derived}, the code $C_k$ is
block-homogeneous.  By Corollary \ref{cor:derived-invariant}, the codes
$\partial^{k+1}H'$ and $\partial^{k+1}H$ are well defined and equivalent to
$\partial C_k$.  This proves the last assertion by induction on $k$.
\end{proof}

For odd $p$, Corollary \ref{cor:odd-invariant} and Lemma \ref{lem:invariance}
show that two kinds of data are equivalence invariants of a reduced code.
The first one is the number of derivations needed to reach a linear code,
and the second one is the sequence of the dimensions of the kernels of the
successive derived codes.  These data determine the parameters of the code.
This gives a second proof of Theorem \ref{thm:main} for odd $p$, which does
not use induction on the length.  For $p=2$, the corresponding data do not
determine the parameters.  In Example \ref{ex:t7}, the codes $H^{4,0}$ and
$H^{2,0,2}$ both reach a linear code after one derivation, and the kernels of
$H^{4,0}$, $H^{2,0,2}$, $H^{4}$ and $H^{2,0}$ have dimensions $5$, $5$, $4$
and $4$.

\begin{example}\label{ex:t11}
The first derivation does not always separate two reduced codes with the
same dimension of the kernel.  Consider the codes $H^{4,0,0}$ and
$H^{2,0,2,0}$.  Both codes have length $p^{11}$, and their kernels have
dimension $4$ if $p$ is odd, and $5$ if $p=2$.  If $p$ is odd, their derived
codes are $H^{4,1}$ and $H^{2,0,3}$.  Both of them are reduced, and both have
kernel dimension $5$.  Deriving once more, we obtain $H^{5}$, which is
linear, and $H^{2,1}$, which is not.  Hence, $H^{4,0,0}$ and $H^{2,0,2,0}$
are nonequivalent.  If $p=2$, their derived codes are $H^{4,0}$ and
$H^{2,0,2}$, which are the two codes of Example \ref{ex:t7}.  Hence, the
proof of Theorem \ref{thm:main} reduces the comparison of $H^{4,0,0}$ and
$H^{2,0,2,0}$ to that of $H^{4,0}$ and $H^{2,0,2}$, which is settled by
Proposition \ref{prop:exceptional}.
\end{example}

\begin{example}\label{ex:family}
For odd $p$, the descent may have to be iterated an arbitrary number of
times before the dimensions of the kernels separate two codes.  For
$k\geq2$, let $H$ be the $\Z_{p^{k+1}}$-linear GH code with parameters
$(3,0,\dots,0)$, and let $H'$ be the $\Z_{p^{k+2}}$-linear GH code with
parameters $(2,0,0,1,0,\dots,0)$, where the last $k-2$ parameters of $H'$ are
zero.  Both codes have length $p^{3k+2}$ and kernel dimension $3$.  By
Proposition \ref{prop:iterated}, the codes $\partial^iH$ and $\partial^iH'$
are reduced for $1\leq i\leq k-1$.  Their kernels both have dimension $4$ for
$1\leq i\leq k-2$, while those of $\partial^{k-1}H$ and $\partial^{k-1}H'$
have dimensions $4$ and $3$, respectively.  Hence, the dimensions of the
kernels of the successive derived codes of $H$ and $H'$ agree up to the
$(k-2)$th step, and they first differ at the $(k-1)$th one.
\end{example}

\begin{corollary}\label{cor:representatives}
Let $t\geq3$.  Then, the linear code $H^{t+1}$ and the nonlinear reduced
codes of length $p^t$ form a complete system of representatives for the
equivalence classes of the $\Z_{p^s}$-linear GH codes of length $p^t$, with
$s\geq1$.  If $p=2$, then the only linear reduced code of length $2^t$ is
$H^{2,t-3}$, and it is equivalent to $H^{t+1}$.
\end{corollary}

\begin{proof}
By \eqref{eq:length}, the only code of length $p^t$ with $s=1$ is $H^{t+1}$.
By Proposition \ref{prop:reduction}, every $\Z_{p^s}$-linear GH code of
length $p^t$ with $s\geq2$ is permutation equivalent to a nonlinear reduced
code of the same length or to $H^{t+1}$.  By Theorem \ref{thm:main}, two
different reduced codes of length $p^t$ are nonequivalent.  Moreover, the code
$H^{t+1}$ is linear.  Hence, by Lemma \ref{lem:invariance}, it is not
equivalent to any nonlinear reduced code.  If $p=2$, then the only linear
reduced code of length $2^t$ is $H^{2,t-3}$, by Theorem \ref{thm:linearity}
and \eqref{eq:length}, and it is permutation equivalent to $H^{t+1}$ by
Proposition \ref{prop:reduction}.
\end{proof}

\begin{corollary}\label{cor:perm-equiv}
Let $H=H^{t_1,\dots,t_s}$ and $H'=H^{t'_1,\dots,t'_{s'}}$ be codes of the same
length $p^t$, with $s,s'\geq2$.  Then, $H$ and $H'$ are equivalent if and only
if they are permutation equivalent.
\end{corollary}

\begin{proof}
One implication is trivial.  Conversely, assume that $H$ and $H'$ are
equivalent.  By Lemma \ref{lem:invariance}, they are both linear or both
nonlinear.  In the first case, both of them are permutation equivalent to
$H^{1,t-1}$ by Proposition \ref{prop:reduction}, and hence to each other.
In the second case, Proposition \ref{prop:reduction} gives reduced codes $C$
and $C'$ of length $p^t$ which are permutation equivalent to $H$ and to
$H'$, respectively.  The codes $C$ and $C'$ are equivalent, so they are equal
by Theorem \ref{thm:main}.  Therefore, $H$ and $H'$ are permutation
equivalent.
\end{proof}

By Corollary \ref{cor:perm-equiv}, the notions of equivalence used in
\cite{HadamardZps}, in \cite{EquivZ2s,ZpsEquivalence} and in this paper give
the same classification of the $\Z_{p^s}$-linear GH codes.

\begin{corollary}\label{cor:conjecture}
Let $H=H^{t_1,\dots,t_s}$ be a nonlinear reduced code with $t_s=0$.  Then,
$H$ is not equivalent to any other code $H^{t'_1,\dots,t'_{s'}}$ of the same
length with $s'\geq2$.
\end{corollary}

\begin{proof}
Let $H'=H^{t'_1,\dots,t'_{s'}}$ be a code of the same length, with $s'\geq2$,
which is equivalent to $H$.  Then, $H'$ is nonlinear by Lemma
\ref{lem:invariance}.  Let $k$ be the integer of Lemma \ref{lem:R} for $H'$.
By Lemma \ref{lem:R}, the code $H'$ is permutation equivalent to a reduced code
$H''$ whose last parameter is $t'_{s'}+k$.  The codes $H''$ and $H$ are
equivalent, and both of them are reduced.  Hence, their parameters are equal by
Theorem \ref{thm:main}.  In particular, $t'_{s'}+k=t_s=0$, and therefore $k=0$.
This means that $H''=H'$, and so $(t'_1,\dots,t'_{s'})=(t_1,\dots,t_s)$.
\end{proof}

Corollary \ref{cor:conjecture} proves the conjecture of
\cite[Section 3]{ZpsEquivalence} on the codes with $t_1\geq2$ and $t_s=0$, for
every prime $p$ and every $t\geq3$, up to one binary exception described below.
In \cite{ZpsEquivalence}, the conjecture was checked for the nonlinear codes of
length $p^t$ with $t\leq11$ when $p=2$, with $t\leq10$ when $p=3$, and with
$t\leq8$ when $p=5$.  For odd $p$, every reduced code is nonlinear, so the
hypothesis of nonlinearity always holds.  For $p=2$, a reduced code with
$t_s=0$ is linear only if $s=2$, $t_1=2$ and $t=3$, by Theorem
\ref{thm:linearity} and \eqref{eq:length}.  The linear code $H^{2,0}$ of length
$2^3$ is equivalent to $H^{1,2}$ by Proposition \ref{prop:reduction}.  Hence,
it is the only exception to the conjecture.

\section{The number of nonequivalent codes}\label{sec:counting}

Corollary \ref{cor:representatives} reduces the computation of $\cA_{t,p}$ to
a counting problem, which turns out to be a classical one.  A partition of a
positive integer $n$ into $k$ parts is a sequence
$(\lambda_1,\dots,\lambda_k)$ of positive integers with
$\lambda_1\geq\cdots\geq\lambda_k$ and $\lambda_1+\cdots+\lambda_k=n$.  We
denote by $\cP(n)$ the number of partitions of $n$, and by $\cP_s(n)$ the
number of partitions of $n$ into parts of size at most $s$, with the
conventions $\cP(0)=\cP_s(0)=1$ and $\cP_s(n)=0$ for $n<0$.

The next lemma attaches a partition to the parameters of a code.  It only
involves \eqref{eq:length}, which does not depend on $p$.  Recall that
$\kappa_j=t_1+\cdots+t_j$ for $1\leq j\leq s$.

\begin{lemma}\label{lem:partition}
Let $t\geq0$ and $s\geq1$.  Then, the map
$(t_1,\dots,t_s)\mapsto(\kappa_s,\kappa_{s-1},\dots,\kappa_1)$ is a bijection
from the solutions of \eqref{eq:length} in nonnegative integers with
$t_1\geq1$ onto the partitions of $t+1$ into exactly $s$ parts.  Moreover,
$t_1\geq2$ if and only if all the parts are at least $2$.
\end{lemma}

\begin{proof}
Since $t_1\geq1$ and $t_2,\dots,t_s\geq0$, we have
$\kappa_s\geq\cdots\geq\kappa_1=t_1\geq1$.  Moreover,
$$
\sum_{j=1}^{s}\kappa_j=\sum_{j=1}^{s}\sum_{i=1}^{j}t_i
=\sum_{i=1}^{s}(s-i+1)t_i=t+1
$$
by \eqref{eq:length}.  Hence, $(\kappa_s,\dots,\kappa_1)$ is a partition of
$t+1$ into exactly $s$ parts.  Conversely, a partition
$(\lambda_1,\dots,\lambda_s)$ of $t+1$ into $s$ parts comes from the unique
solution given by $t_1=\lambda_s$ and $t_j=\lambda_{s-j+1}-\lambda_{s-j+2}$
for $2\leq j\leq s$.  The smallest part is $\kappa_1=t_1$, which proves the
last assertion.
\end{proof}

For instance, the parameters $(2,0,1,0)$ of a code of length $p^9$ give
$(\kappa_1,\kappa_2,\kappa_3,\kappa_4)=(2,2,3,3)$, and hence the partition
$(3,3,2,2)$ of $10$.  Conversely, the partition $(5,3,2)$ of $10$ comes from
the parameters $(2,1,2)$.

\begin{proposition}\label{prop:X}
Let $t\geq1$ and $s\geq1$.  Then, $X_{t,s,p}=\cP_s(t+1-s)$, and the number
of solutions of \eqref{eq:length} in nonnegative integers with $t_1\geq2$ is
$\cP_s(t+1-2s)$.  Consequently, for $s\geq2$, we have
$\widetilde X_{t,s,p}=\cP_s(t+1-2s)$, except for $p=2$ and $s=2$, where
$\widetilde X_{t,2,2}=\cP_2(t-3)-1$ if $t\geq3$.
\end{proposition}

\begin{proof}
By Lemma \ref{lem:partition}, $X_{t,s,p}$ is the number of partitions of $t+1$
into exactly $s$ parts.  Subtracting $1$ from every part, and deleting the
parts which become zero, is a bijection from these partitions onto the
partitions of $t+1-s$ into at most $s$ parts.  By conjugation
\cite{Andrews}, the latter are as many as the partitions of $t+1-s$ into
parts of size at most $s$.  Hence, $X_{t,s,p}=\cP_s(t+1-s)$.  Similarly, by
Lemma \ref{lem:partition}, the solutions with $t_1\geq2$ correspond to the
partitions of $t+1$ into exactly $s$ parts, all of them at least $2$.
Subtracting $2$ from every part and arguing in the same way, we obtain
$\cP_s(t+1-2s)$ such solutions.  By Theorem \ref{thm:linearity}, all the
corresponding codes are nonlinear, except $H^{2,t-3}$ when $p=2$ and $s=2$.
This proves the last assertion.
\end{proof}

For $t\geq1$ and a partition $\lambda$ of $t+1$ with all parts at least $2$,
we denote by $H_\lambda$ the code $H^{t_1,\dots,t_s}$ of length $p^t$ which
corresponds to $\lambda$ by Lemma \ref{lem:partition}.  The dimension of the
kernel and the derived code of $H_\lambda$ have a simple description in terms
of $\lambda$.  The Young diagram of a partition
$\lambda=(\lambda_1,\dots,\lambda_s)$ is formed by $s$ rows of boxes, aligned
on the left, with $\lambda_i$ boxes in the $i$th row.

\begin{remark}\label{rem:young}
Let $\lambda=(\lambda_1,\dots,\lambda_s)$ be a partition of $t+1$ into
$s\geq2$ parts, all of them at least $2$, and assume that $H_\lambda$ is
nonlinear.  By Lemma \ref{lem:partition}, we have $\lambda_1=\kappa_s=m$.
Hence, by Theorem \ref{thm:kernel}, the kernel of $H_\lambda$ has dimension
$\lambda_1$ if $p$ is odd, and $\lambda_1+1$ if $p=2$.  Moreover, by Corollary
\ref{cor:derived} and Lemma \ref{lem:partition}, the derived code of
$H_\lambda$ is, up to equivalence, the code $H_{\lambda'}$, where
$$
\lambda'=\begin{cases}
(\lambda_2+1,\lambda_3,\dots,\lambda_s) & \text{if $p$ is odd},\\
(\lambda_2,\lambda_3,\dots,\lambda_s) & \text{if } p=2.
\end{cases}
$$
That is, the Young diagram of $\lambda'$ is obtained from that of $\lambda$ by
deleting the first row and, if $p$ is odd, by adding one box to the new first
row.  Figure \ref{fig:young} shows the successive derived codes of the two
codes of Example \ref{ex:t11}.  For odd $p$, the first rows of the successive
diagrams have $\lambda_1,\lambda_2+1,\dots,\lambda_s+1$ boxes.  By Theorem
\ref{thm:kernel} and Proposition \ref{prop:iterated}, these numbers are the
dimensions of the kernels of
$H_\lambda,\partial H_\lambda,\dots,\partial^{s-1}H_\lambda$, and they
determine $\lambda$, as in Corollary \ref{cor:odd-invariant}.  For $p=2$ and
$n\geq4$, the partitions $(n)$ and $(n-2,2)$ give the linear codes $H^{n}$ and
$H^{2,n-4}$, which are equivalent by Proposition \ref{prop:reduction}.  In
particular, the codes $H^{a_1,a_2}$ and $H^{2,b_2,b_3}$ of Proposition
\ref{prop:exceptional} have the equivalent derived codes $H^{a_1}$ and
$H^{2,a_1-4}$, and the descent does not separate them.
\end{remark}

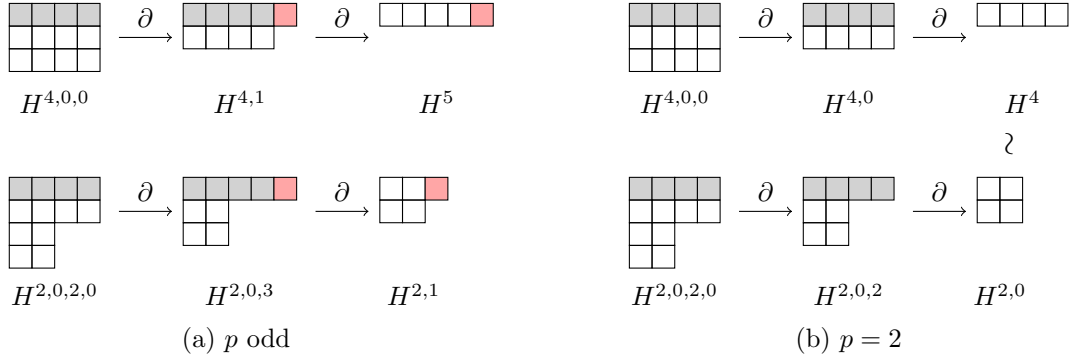
\begin{figure}[!ht]
\centering
\newcommand{\yd}[4]{%
  \begin{scope}[shift={#1}]
  \foreach \len [count=\i from 0] in {#2}{
    \foreach \j in {1,...,\len}{
      \ifnum\i=0
        \ifnum#3=1 \fill[gray!35] (\j*0.3-0.3,-\i*0.3) rectangle ++(0.3,-0.3);\fi
        \ifnum#4=\j \fill[red!35] (\j*0.3-0.3,-\i*0.3) rectangle ++(0.3,-0.3);\fi
      \fi
      \draw (\j*0.3-0.3,-\i*0.3) rectangle ++(0.3,-0.3);
    }
  }
  \end{scope}}
\begin{tikzpicture}[font=\small]
\yd{(0,0)}{4,4,4}{1}{0}
\draw[->] (1.45,-0.45) -- node[above] {$\partial$} (2.15,-0.45);
\yd{(2.3,0)}{5,4}{1}{5}
\draw[->] (4.05,-0.45) -- node[above] {$\partial$} (4.75,-0.45);
\yd{(4.9,0)}{5}{0}{5}
\node at (0.6,-1.3) {$H^{4,0,0}$};
\node at (3.05,-1.3) {$H^{4,1}$};
\node at (5.65,-1.3) {$H^{5}$};
\begin{scope}[yshift=-2.3cm]
\yd{(0,0)}{4,4,2,2}{1}{0}
\draw[->] (1.45,-0.45) -- node[above] {$\partial$} (2.15,-0.45);
\yd{(2.3,0)}{5,2,2}{1}{5}
\draw[->] (4.05,-0.45) -- node[above] {$\partial$} (4.75,-0.45);
\yd{(4.9,0)}{3,2}{0}{3}
\node at (0.6,-1.55) {$H^{2,0,2,0}$};
\node at (3.05,-1.55) {$H^{2,0,3}$};
\node at (5.35,-1.55) {$H^{2,1}$};
\end{scope}
\node at (3.0,-4.45) {(a) $p$ odd};
\begin{scope}[xshift=8.2cm]
\yd{(0,0)}{4,4,4}{1}{0}
\draw[->] (1.45,-0.45) -- node[above] {$\partial$} (2.15,-0.45);
\yd{(2.3,0)}{4,4}{1}{0}
\draw[->] (3.75,-0.45) -- node[above] {$\partial$} (4.45,-0.45);
\yd{(4.6,0)}{4}{0}{0}
\node at (0.6,-1.3) {$H^{4,0,0}$};
\node at (2.9,-1.3) {$H^{4,0}$};
\node at (5.2,-1.3) {$H^{4}$};
\begin{scope}[yshift=-2.3cm]
\yd{(0,0)}{4,4,2,2}{1}{0}
\draw[->] (1.45,-0.45) -- node[above] {$\partial$} (2.15,-0.45);
\yd{(2.3,0)}{4,2,2}{1}{0}
\draw[->] (3.75,-0.45) -- node[above] {$\partial$} (4.45,-0.45);
\yd{(4.6,0)}{2,2}{0}{0}
\node at (0.6,-1.55) {$H^{2,0,2,0}$};
\node at (2.9,-1.55) {$H^{2,0,2}$};
\node at (4.9,-1.55) {$H^{2,0}$};
\end{scope}
\node[rotate=90,font=\large] at (5.05,-1.85) {$\sim$};
\node at (2.9,-4.45) {(b) $p=2$};
\end{scope}
\end{tikzpicture}
\caption{The successive derived codes of $H^{4,0,0}$ and $H^{2,0,2,0}$ on
Young diagrams.  At each step, the gray row is deleted and, for odd $p$, the
red box is added.  Codes joined by $\sim$ are equivalent.}
\label{fig:young}
\end{figure}

\begin{theorem}\label{thm:count}
Let $t\geq3$.  Then, every $\Z_{p^s}$-linear GH code of length $p^t$ with
$s\geq2$ is equivalent to some $H_\lambda$.  If $p$ is odd, the codes
$H_\lambda$ are pairwise nonequivalent.  If $p=2$, two different codes
$H_\lambda$ and $H_\nu$ are equivalent if and only if
$\{\lambda,\nu\}=\{(t+1),(t-1,2)\}$.  Consequently,
\begin{equation}\label{eq:count}
\cA_{t,p}=1+\sum_{s=2}^{\lfloor(t+1)/2\rfloor}\widetilde X_{t,s,p}
=\begin{cases}
\cP(t+1)-\cP(t) & \text{if $p$ is odd},\\
\cP(t+1)-\cP(t)-1 & \text{if } p=2,
\end{cases}
\end{equation}
and the upper bound \eqref{eq:old-bound} is tight.
\end{theorem}

\begin{proof}
By Lemma \ref{lem:partition}, the map $\lambda\mapsto H_\lambda$ is a
bijection from the set of partitions of $t+1$ with all parts at least $2$
onto the set of codes $H^{t_1,\dots,t_s}$ of length $p^t$ with $s\geq1$ and
$t_1\geq2$.  The partition $(t+1)$ corresponds to $H^{t+1}$, and the
partitions into $s\geq2$ parts correspond to the reduced $\Z_{p^s}$-linear GH
codes.

By Corollary \ref{cor:representatives}, every $\Z_{p^s}$-linear GH code of
length $p^t$ with $s\geq2$ is equivalent to $H^{t+1}$ or to a reduced code,
that is, to some $H_\lambda$.  Moreover, $H^{t+1}$ is permutation equivalent to
$H^{1,t-1}$ by Proposition \ref{prop:reduction}.  Hence, every code $H_\lambda$
is equivalent to a code with $s\geq2$.  Now, let $\lambda\neq\nu$ be partitions
such that $H_\lambda$ and $H_\nu$ are equivalent.  If both of them have at
least two parts, then $H_\lambda$ and $H_\nu$ are reduced codes with different
parameters, which contradicts Theorem \ref{thm:main}.  Hence, we may assume
that $\nu=(t+1)$.  Then, $H_\nu=H^{t+1}$ is linear, and so is $H_\lambda$ by
Lemma \ref{lem:invariance}.  Since $\lambda$ has at least two parts,
$H_\lambda$ is a linear reduced code.  If $p$ is odd, this is impossible.  If
$p=2$, then $H_\lambda=H^{2,t-3}$ by Theorem \ref{thm:linearity}, and
$\lambda=(t-1,2)$ by Lemma \ref{lem:partition}.  Conversely, if $p=2$, the
codes $H^{2,t-3}$ and $H^{t+1}$ are linear, and they are equivalent by
Proposition \ref{prop:reduction}.  This proves the first three assertions.

Deleting one part equal to $1$ is a bijection from the partitions of $t+1$
having a part equal to $1$ onto the partitions of $t$.  Hence, $t+1$ has
$\cP(t)$ partitions with a part equal to $1$, and $\cP(t+1)-\cP(t)$ partitions
with all parts at least $2$.  By the first three assertions, $\cA_{t,p}$ is
equal to this number if $p$ is odd, and to this number minus $1$ if $p=2$.
Finally, by Corollary \ref{cor:representatives}, the equivalence classes are
represented by $H^{t+1}$ and by the nonlinear reduced codes of length $p^t$,
one code for each class.  A reduced $\Z_{p^s}$-linear GH code of length $p^t$
corresponds to a partition of $t+1$ into $s$ parts, all of them at least $2$,
and such a partition exists only if $2s\leq t+1$.  Hence,
$\cA_{t,p}=1+\sum_{s=2}^{\lfloor(t+1)/2\rfloor}\widetilde X_{t,s,p}$, which is
the right-hand side of \eqref{eq:old-bound}.
\end{proof}

Figure \ref{fig:classes} shows the equivalence classes of the codes of length
$p^6$.  For odd $p$, there are $\cP(7)-\cP(6)=4$ classes.  For $p=2$, the
classes of $H^{7}$ and $H^{2,3}$ coincide, and there are $3$ classes.

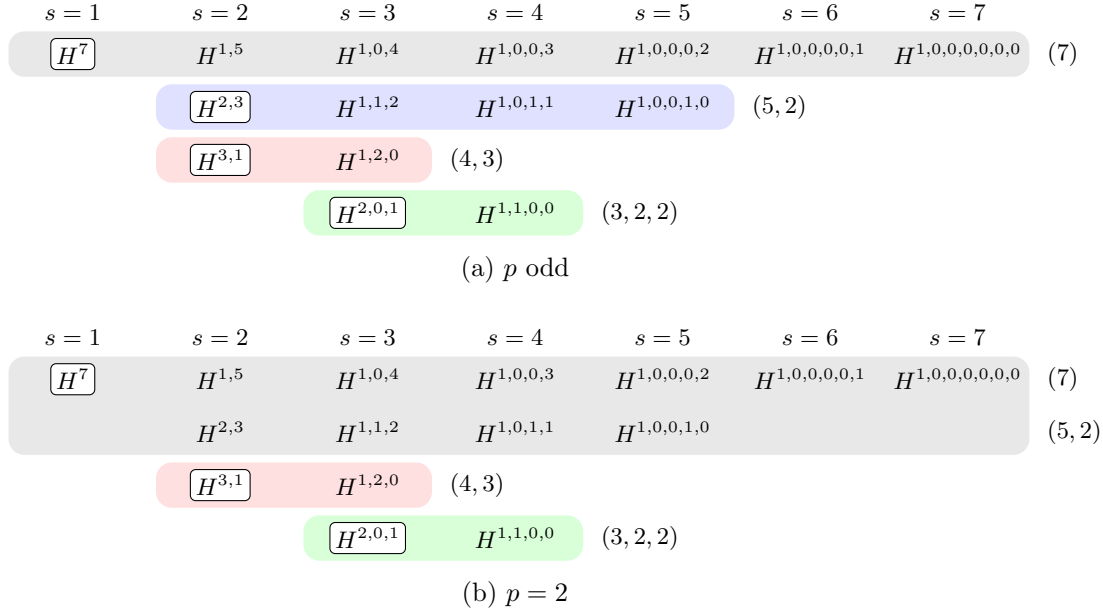
\begin{figure}[!ht]
\centering
\begin{tikzpicture}[font=\footnotesize,
  rep/.style={draw,rounded corners=2pt,inner sep=2pt,fill=white}]
\def\dx{1.95}
\foreach \s in {1,...,7}{\node at ({(\s-1)*\dx},0.55) {$s=\s$};}
\fill[gray!18,rounded corners=6pt] (-0.85,-0.3) rectangle ({6*\dx+0.95},0.3);
\fill[blue!12,rounded corners=6pt] ({\dx-0.85},-1.0) rectangle
  ({4*\dx+0.95},-0.4);
\fill[red!12,rounded corners=6pt] ({\dx-0.85},-1.7) rectangle
  ({2*\dx+0.85},-1.1);
\fill[green!15,rounded corners=6pt] ({2*\dx-0.85},-2.4) rectangle
  ({3*\dx+0.9},-1.8);
\node[rep] at (0,0) {$H^{7}$};
\node at ({1*\dx},0) {$H^{1,5}$};
\node at ({2*\dx},0) {$H^{1,0,4}$};
\node at ({3*\dx},0) {$H^{1,0,0,3}$};
\node at ({4*\dx},0) {$H^{1,0,0,0,2}$};
\node at ({5*\dx},0) {$H^{1,0,0,0,0,1}$};
\node at ({6*\dx},0) {$H^{1,0,0,0,0,0,0}$};
\node[rep] at ({1*\dx},-0.7) {$H^{2,3}$};
\node at ({2*\dx},-0.7) {$H^{1,1,2}$};
\node at ({3*\dx},-0.7) {$H^{1,0,1,1}$};
\node at ({4*\dx},-0.7) {$H^{1,0,0,1,0}$};
\node[rep] at ({1*\dx},-1.4) {$H^{3,1}$};
\node at ({2*\dx},-1.4) {$H^{1,2,0}$};
\node[rep] at ({2*\dx},-2.1) {$H^{2,0,1}$};
\node at ({3*\dx},-2.1) {$H^{1,1,0,0}$};
\node[anchor=west] at ({6*\dx+1.05},0) {$(7)$};
\node[anchor=west] at ({4*\dx+1.05},-0.7) {$(5,2)$};
\node[anchor=west] at ({2*\dx+0.95},-1.4) {$(4,3)$};
\node[anchor=west] at ({3*\dx+1.0},-2.1) {$(3,2,2)$};
\node[font=\small] at ({3*\dx},-2.85) {(a) $p$ odd};
\begin{scope}[yshift=-4.3cm]
\foreach \s in {1,...,7}{\node at ({(\s-1)*\dx},0.55) {$s=\s$};}
\fill[gray!18,rounded corners=6pt] (-0.85,-1.0) rectangle ({6*\dx+0.95},0.3);
\fill[red!12,rounded corners=6pt] ({\dx-0.85},-1.7) rectangle
  ({2*\dx+0.85},-1.1);
\fill[green!15,rounded corners=6pt] ({2*\dx-0.85},-2.4) rectangle
  ({3*\dx+0.9},-1.8);
\node[rep] at (0,0) {$H^{7}$};
\node at ({1*\dx},0) {$H^{1,5}$};
\node at ({2*\dx},0) {$H^{1,0,4}$};
\node at ({3*\dx},0) {$H^{1,0,0,3}$};
\node at ({4*\dx},0) {$H^{1,0,0,0,2}$};
\node at ({5*\dx},0) {$H^{1,0,0,0,0,1}$};
\node at ({6*\dx},0) {$H^{1,0,0,0,0,0,0}$};
\node at ({1*\dx},-0.7) {$H^{2,3}$};
\node at ({2*\dx},-0.7) {$H^{1,1,2}$};
\node at ({3*\dx},-0.7) {$H^{1,0,1,1}$};
\node at ({4*\dx},-0.7) {$H^{1,0,0,1,0}$};
\node[rep] at ({1*\dx},-1.4) {$H^{3,1}$};
\node at ({2*\dx},-1.4) {$H^{1,2,0}$};
\node[rep] at ({2*\dx},-2.1) {$H^{2,0,1}$};
\node at ({3*\dx},-2.1) {$H^{1,1,0,0}$};
\node[anchor=west] at ({6*\dx+1.05},0) {$(7)$};
\node[anchor=west] at ({6*\dx+1.05},-0.7) {$(5,2)$};
\node[anchor=west] at ({2*\dx+0.95},-1.4) {$(4,3)$};
\node[anchor=west] at ({3*\dx+1.0},-2.1) {$(3,2,2)$};
\node[font=\small] at ({3*\dx},-2.85) {(b) $p=2$};
\end{scope}
\end{tikzpicture}
\caption{The equivalence classes of the codes $H^{t_1,\dots,t_s}$ of length
$p^6$, placed according to $s$.  Each shaded region is a class, and the
framed codes are the representatives of Corollary \ref{cor:representatives}.
The partition $\lambda$ of the code $H_\lambda$ of each row is given on the
right.}
\label{fig:classes}
\end{figure}

As recalled in Section \ref{sec:intro}, every binary Hadamard code which is
the Gray map image of a $\Z_{2^s}$-additive code is equivalent to a code
$H^{t_1,\dots,t_s}$.  Hence, for $p=2$, Theorem \ref{thm:count} counts all
these Hadamard codes of length $2^t$ up to equivalence, and not only the
codes $H^{t_1,\dots,t_s}$.  We now fix $s$.  Recall that $X_{t,s,p}$ is the
number of codes $H^{t_1,\dots,t_s}$ of length $p^t$.

\begin{theorem}\label{thm:count-fixed-s}
Let $t\geq3$ and $s\geq2$.  Then, two codes $H^{t_1,\dots,t_s}$ and
$H^{t'_1,\dots,t'_s}$ of length $p^t$ are equivalent if and only if
$t_i=t'_i$ for $1\leq i\leq s$, or $p=2$ and both codes are linear.
Consequently, if $p$ is odd, then
\begin{equation}\label{eq:Ats}
\cA_{t,s,p}=X_{t,s,p}=\cP_s(t+1-s),
\end{equation}
and, if $p=2$, then
\begin{equation}\label{eq:Ats2}
\cA_{t,s,2}=\begin{cases}
X_{t,s,2}-1=\cP_s(t+1-s)-1 & \text{if } 2\leq s\leq t-1,\\
1 & \text{if } s\in\{t,t+1\},\\
0 & \text{if } s\geq t+2.
\end{cases}
\end{equation}
\end{theorem}

\begin{proof}
If $t_i=t'_i$ for $1\leq i\leq s$, then the two codes coincide.  If $p=2$
and both codes are linear, then they are permutation equivalent to $H^{t+1}$
by Proposition \ref{prop:reduction}.  Conversely, assume that the two codes
are equivalent.  If one of them is linear, then so is the other one by Lemma
\ref{lem:invariance}.  If, moreover, $p$ is odd, then both codes have the
parameters $(1,0,\dots,0,t+1-s)$ by Theorem \ref{thm:linearity} and
\eqref{eq:length}.

Assume now that both codes are nonlinear.  Let $k$ and $k'$ be the integers
of Lemma \ref{lem:R} for the two codes.  By Lemma \ref{lem:R}, the two codes
are permutation equivalent to the reduced codes with parameters
$R^k(t_1,\dots,t_s)$ and $R^{k'}(t'_1,\dots,t'_s)$, which have $s-k$ and
$s-k'$ parameters, respectively.  These two reduced codes are equivalent,
since each of them is permutation equivalent to one of the original codes.  By
Theorem \ref{thm:main}, their parameters are equal.  In particular, $s-k=s-k'$,
that is, $k=k'$.  Since $R$ is injective, so is $R^k$.  Therefore,
$(t_1,\dots,t_s)=(t'_1,\dots,t'_s)$.

We now count.  If $p$ is odd, the first part shows that the $X_{t,s,p}$
codes $H^{t_1,\dots,t_s}$ of length $p^t$ are pairwise nonequivalent.
Hence, $\cA_{t,s,p}=X_{t,s,p}$, and $X_{t,s,p}=\cP_s(t+1-s)$ by Proposition
\ref{prop:X}.  Let $p=2$, and let $\ell$ be the number of linear codes
$H^{t_1,\dots,t_s}$ of length $2^t$.  By the first part,
$\cA_{t,s,2}=X_{t,s,2}-\ell+1$ if $\ell\geq1$, and $\cA_{t,s,2}=X_{t,s,2}$ if
$\ell=0$.  For $s=2$, the linear codes are $H^{1,t-1}$ and $H^{2,t-3}$ by
Theorem \ref{thm:linearity}, so $\ell=2$.  For $s\geq3$, their parameters are
$(1,0,\dots,0,t_s)$ or $(1,0,\dots,0,1,t_s)$, which are different.  By
\eqref{eq:length}, the first one requires $t_s=t+1-s$, and the second one
requires $t_s=t-1-s$.  Hence, $\ell=2$ if $s\leq t-1$, $\ell=1$ if
$s\in\{t,t+1\}$, and $\ell=0$ if $s\geq t+2$.  Finally,
$X_{t,s,2}=\cP_s(t+1-s)$ is equal to $1$ if $s\in\{t,t+1\}$, and to $0$ if
$s\geq t+2$.  This gives \eqref{eq:Ats2}.
\end{proof}

By Theorem \ref{thm:count-fixed-s}, the upper bounds
$\cA_{t,s,p}\leq X_{t,s,p}$ for odd $p$ \cite[Theorem 6]{HadamardZps} and
$\cA_{t,s,2}\leq X_{t,s,2}-1$ \cite{KernelZ2s} are tight for every $t\geq3$
and every $s$ with $2\leq s\leq t-1$.  This proves, for every prime $p$, the
equality conjectured in \cite[Section 7]{HadamardZps}.  For $s=2$, we have
$\cP_2(n)=\lfloor n/2\rfloor+1$ for $n\geq0$.  Hence, \eqref{eq:Ats} and
\eqref{eq:Ats2} give $\cA_{t,2,p}=\lfloor(t+1)/2\rfloor$ for odd $p$ and
$\cA_{t,2,2}=\lfloor(t-1)/2\rfloor$, which are the values obtained in
\cite{HadamardZps} and \cite{Kro:2001:Z4_Had_Perf,PRV06}.  For odd $p$ and
$s\geq t$, \eqref{eq:Ats} gives $\cA_{t,s,p}=1$ for $s\in\{t,t+1\}$ and
$\cA_{t,s,p}=0$ for $s\geq t+2$, in accordance with
\cite[Theorem 5]{HadamardZps}.

Table \ref{tab:Ats} gives the values of $\cA_{t,s,p}$ for odd $p$, obtained
from \eqref{eq:Ats}, and the values of $\cA_{t,p}$ and $\cA_{t,2}$, obtained
from \eqref{eq:count}, for $3\leq t\leq14$.  Empty entries are zero.  For
$p=2$, the values of $\cA_{t,s,2}$ are obtained by subtracting $1$ from the
entries with $2\leq s\leq t-1$, by \eqref{eq:Ats2}.  The last two rows are
not the column sums of the entries above them, since codes with different
values of $s$ may be equivalent.  For $3\leq t\leq11$ if $p=2$, for
$3\leq t\leq10$ if $p=3$, and for $3\leq t\leq8$ if $p=5$, the values in
Table \ref{tab:Ats} agree with the ones obtained in
\cite{KernelZ2s,EquivZ2s,HadamardZps,ZpsEquivalence} by computing the rank
with \textsc{Magma}.  For $p\geq7$, the results of
\cite{HadamardZps,ZpsEquivalence} determine $\cA_{t,p}$ only for $t\leq6$,
where the dimension of the kernel suffices by Example \ref{ex:t7}.

Besides \eqref{eq:old-bound}, two other upper bounds on $\cA_{t,p}$ are
known, namely
$$
\cA_{t,p}\leq1+\sum_{s=2}^{\lfloor(t+1)/2\rfloor}(\cA_{t,s,p}-1)
\quad\text{and}\quad
\cA_{t,p}\leq1+\sum_{s=2}^{t-1}(\cA_{t,s,p}-1),
$$
which are given in \cite{EquivZ2s,ZpsEquivalence} and in
\cite{KernelZ2s,HadamardZps}, respectively.  For $p=2$, the second sum stops
at $s=t-2$ in \cite{KernelZ2s}, which gives the same value, since
$\cA_{t,t-1,2}=1$ by \eqref{eq:Ats2}.  These bounds are far from the exact
value.  For instance, evaluated with \eqref{eq:Ats} and \eqref{eq:Ats2} for
$t=20$, they give $644$ and $772$ for odd $p$, and $635$ and $754$ for
$p=2$, while $\cA_{20,p}=165$ and $\cA_{20,2}=164$.

\begin{table}[!ht]
\small
\setlength{\tabcolsep}{4.5pt}
\renewcommand{\arraystretch}{1.15}
\centering
\caption{The numbers $\cA_{t,s,p}$ for odd $p$, and the numbers $\cA_{t,p}$
and $\cA_{t,2}$, for $3\leq t\leq14$.}
\label{tab:Ats}
\vspace{1.5mm}
\begin{tabular}{|c||c|c|c|c|c|c|c|c|c|c|c|c|}
\hline
$t$ & $3$ & $4$ & $5$ & $6$ & $7$ & $8$ & $9$ & $10$ & $11$ & $12$ & $13$
& $14$\\
\hline\hline
$s=2$  & 2 & 2 & 3 & 3 & 4 & 4 & 5 & 5 & 6 & 6 & 7 & 7\\ \hline
$s=3$  & 1 & 2 & 3 & 4 & 5 & 7 & 8 & 10 & 12 & 14 & 16 & 19\\ \hline
$s=4$  & 1 & 1 & 2 & 3 & 5 & 6 & 9 & 11 & 15 & 18 & 23 & 27\\ \hline
$s=5$  &   & 1 & 1 & 2 & 3 & 5 & 7 & 10 & 13 & 18 & 23 & 30\\ \hline
$s=6$  &   &   & 1 & 1 & 2 & 3 & 5 & 7 & 11 & 14 & 20 & 26\\ \hline
$s=7$  &   &   &   & 1 & 1 & 2 & 3 & 5 & 7 & 11 & 15 & 21\\ \hline
$s=8$  &   &   &   &   & 1 & 1 & 2 & 3 & 5 & 7 & 11 & 15\\ \hline
$s=9$  &   &   &   &   &   & 1 & 1 & 2 & 3 & 5 & 7 & 11\\ \hline
$s=10$ &   &   &   &   &   &   & 1 & 1 & 2 & 3 & 5 & 7\\ \hline
$s=11$ &   &   &   &   &   &   &   & 1 & 1 & 2 & 3 & 5\\ \hline
$s=12$ &   &   &   &   &   &   &   &   & 1 & 1 & 2 & 3\\ \hline
$s=13$ &   &   &   &   &   &   &   &   &   & 1 & 1 & 2\\ \hline
$s=14$ &   &   &   &   &   &   &   &   &   &   & 1 & 1\\ \hline
$s=15$ &   &   &   &   &   &   &   &   &   &   &   & 1\\ \hline\hline
$\cA_{t,p}$, $p$ odd & 2 & 2 & 4 & 4 & 7 & 8 & 12 & 14 & 21 & 24 & 34 & 41\\
\hline
$\cA_{t,2}$ & 1 & 1 & 3 & 3 & 6 & 7 & 11 & 13 & 20 & 23 & 33 & 40\\
\hline
\end{tabular}
\end{table}

Table \ref{tab:t9} illustrates Theorems \ref{thm:main} and \ref{thm:count}
for $t=9$.  There are $\cP(10)-\cP(9)=42-30=12$ partitions of $10$ with all
parts at least $2$.  Hence, $\cA_{9,p}=12$ for odd $p$, and $\cA_{9,2}=11$.
For each of these partitions $\lambda$, the table gives the code
$H=H_\lambda$, the dimension $\kappa$ of its kernel and its derived code
$\partial H$, for odd $p$ and for $p=2$.  A dash means that $H$ is linear.
In brackets, we give the dimension of the kernel of $\partial H$ when
$\partial H$ is nonlinear, and the derived codes without brackets are linear.
For $p=2$, the first two codes are linear, and they are equivalent.  For
odd $p$, the three codes $H^{2,2,0}$, $H^{3,0,1}$ and $H^{2,0,0,2}$ have
kernel dimension $4$.  Their derived codes have kernel dimensions $5$, $4$
and $3$, so the first derivation separates them.  The same happens for the
two codes with kernel dimension $6$ and for the two codes with kernel
dimension $5$, since exactly one of the two derived codes is linear in each
case.  For $p=2$, the same three codes have kernel dimension $5$, and their
derived codes are $H^{2,2}$, which is linear, and $H^{3,0}$ and $H^{2,0,0}$,
whose kernels have dimensions $4$ and $3$.  The pairs formed by $H^{4,2}$ and
$H^{2,0,4}$, and by $H^{5,0}$ and $H^{2,1,2}$, have kernel dimensions $7$ and
$6$, respectively, and their derived codes are linear.  These two pairs are
instances of Proposition \ref{prop:exceptional}.  By \eqref{eq:rank4} and
\eqref{eq:rank8}, the ranks are $13$ and $12$ for the first pair, and $16$
and $14$ for the second one.

\begin{table}[!ht]
\small
\setlength{\tabcolsep}{6pt}
\renewcommand{\arraystretch}{1.15}
\centering
\caption{The codes $H=H_\lambda$ of length $p^{9}$, the dimensions of their
kernels and their derived codes.}
\label{tab:t9}
\vspace{1.5mm}
\begin{tabular}{|c|c|c||c|c||c|c|}
\hline
\multirow{2}{*}{partition} & \multirow{2}{*}{code} & \multirow{2}{*}{$s$}
& \multicolumn{2}{c||}{$p$ odd} & \multicolumn{2}{c|}{$p=2$}\\
\cline{4-7}
& & & $\kappa$ & $\partial H$ & $\kappa$ & $\partial H$\\
\hline\hline
$(10)$ & $H^{10}$ & $1$ & $10$ & -- & $10$ & --\\ \hline
$(8,2)$ & $H^{2,6}$ & $2$ & $8$ & $H^{3}$ & $10$ & --\\ \hline
$(7,3)$ & $H^{3,4}$ & $2$ & $7$ & $H^{4}$ & $8$ & $H^{3}$\\ \hline
$(6,4)$ & $H^{4,2}$ & $2$ & $6$ & $H^{5}$ & $7$ & $H^{4}$\\ \hline
$(5,5)$ & $H^{5,0}$ & $2$ & $5$ & $H^{6}$ & $6$ & $H^{5}$\\ \hline
$(6,2,2)$ & $H^{2,0,4}$ & $3$ & $6$ & $H^{2,1}\,[3]$ & $7$ & $H^{2,0}$\\
\hline
$(5,3,2)$ & $H^{2,1,2}$ & $3$ & $5$ & $H^{2,2}\,[4]$ & $6$ & $H^{2,1}$\\
\hline
$(4,4,2)$ & $H^{2,2,0}$ & $3$ & $4$ & $H^{2,3}\,[5]$ & $5$ & $H^{2,2}$\\
\hline
$(4,3,3)$ & $H^{3,0,1}$ & $3$ & $4$ & $H^{3,1}\,[4]$ & $5$
& $H^{3,0}\,[4]$\\ \hline
$(4,2,2,2)$ & $H^{2,0,0,2}$ & $4$ & $4$ & $H^{2,0,1}\,[3]$ & $5$
& $H^{2,0,0}\,[3]$\\ \hline
$(3,3,2,2)$ & $H^{2,0,1,0}$ & $4$ & $3$ & $H^{2,0,2}\,[4]$ & $4$
& $H^{2,0,1}\,[4]$\\ \hline
$(2,2,2,2,2)$ & $H^{2,0,0,0,0}$ & $5$ & $2$ & $H^{2,0,0,1}\,[3]$ & $3$
& $H^{2,0,0,0}\,[3]$\\ \hline
\end{tabular}
\end{table}

\FloatBarrier

\section{Conclusions and further research}\label{sec:conclusions}

We have obtained the full classification of the $\Z_{p^s}$-linear GH codes for
every prime $p$.  By Theorem \ref{thm:main}, two reduced codes of the same
length are equivalent if and only if their parameters are equal.  By Theorem
\ref{thm:count}, for odd $p$ and $t\geq3$, the equivalence classes of the codes
of length $p^t$ are in bijection with the partitions of $t+1$ with all parts at
least $2$.  For $p=2$, the same holds after identifying the partitions $(t+1)$
and $(t-1,2)$.  Hence, $\cA_{t,p}=\cP(t+1)-\cP(t)$ for odd $p$, and
$\cA_{t,2}=\cP(t+1)-\cP(t)-1$.  For a fixed $s$, Theorem
\ref{thm:count-fixed-s} gives $\cA_{t,s,p}=X_{t,s,p}$ for odd $p$, and
$\cA_{t,s,2}=X_{t,s,2}-1$ for $2\leq s\leq t-1$.

The method is the same for every prime, namely the derived code obtained by
puncturing outside a kernel block, but it behaves differently for $p=2$ and for
odd $p$.  Table \ref{tab:compare} summarizes the differences.  In this table,
$H=H^{t_1,\dots,t_s}=H_\lambda$ is a nonlinear reduced code of length $p^t$,
where $\lambda=(\lambda_1,\dots,\lambda_s)$, and $m=t_1+\cdots+t_s$.  For odd
$p$, the descent can always be iterated until $s$ becomes $1$.  The parameters
are then recovered from the dimensions of the kernels of the successive derived
codes, and no rank computation is needed.  For $p=2$, the known formulas for
the rank of the $\Z_4$-linear and the $\Z_8$-linear Hadamard codes are still
needed in one comparison.

\begin{table}[!ht]
\small
\setlength{\tabcolsep}{6pt}
\renewcommand{\arraystretch}{1.2}
\centering
\caption{The main differences between the cases $p=2$ and $p$ odd.}
\label{tab:compare}
\vspace{1.5mm}
\begin{tabular}{|l|l|l|}
\hline
& $p=2$ & $p$ odd\\
\hline\hline
linear reduced codes & $H^{2,t_2}$ & none\\ \hline
dimension of the kernel of $H$ & $m+1$ & $m$\\ \hline
label of the coordinate $(\zz,h)$ & $(\chi(\zz),\wt_H(h)\bmod 2)$
& $\chi(\zz)$\\ \hline
kernel blocks & each fiber meets two blocks & unions of fibers\\ \hline
permutations used & $\pi_\delta$ and $\theta$ & $\pi_\delta$\\ \hline
derived code $\partial H$ & $H^{t_1,\dots,t_{s-1}}$
& $H^{t_1,\dots,t_{s-2},t_{s-1}+1}$\\ \hline
partition of $\partial H$ & $(\lambda_2,\dots,\lambda_s)$
& $(\lambda_2+1,\lambda_3,\dots,\lambda_s)$\\ \hline
rank needed & for $H^{a_1,a_2}$ and $H^{2,b_2,b_3}$ & never\\ \hline
number $\cA_{t,p}$ of classes, $t\geq3$ & $\cP(t+1)-\cP(t)-1$
& $\cP(t+1)-\cP(t)$\\
\hline
\end{tabular}
\end{table}

We end with three questions for further research.  First, formulas for the
rank of the codes $H^{t_1,\dots,t_s}$ are only known for $p=2$ and
$s\in\{2,3\}$ \cite{PRV06,fernandez2019mathbb}.  General formulas would decide
the conjectures of \cite[Section 7]{HadamardZps} and
\cite[Section 3]{ZpsEquivalence}.  They state that the rank classifies the
codes $H^{t_1,\dots,t_s}$ of a given length for a fixed $s$, and that the rank
and the dimension of the kernel classify all the codes of a given length.
Second, for $p=2$, every Hadamard code which is the Gray map image of a
$\Z_{2^s}$-additive code is equivalent to a code $H^{t_1,\dots,t_s}$
\cite{Krotov:2007,KernelZ2s}.  It is not known whether the same holds for odd
$p$.  A positive answer would extend our classification to every GH code which
is the Gray map image of a $\Z_{p^s}$-additive code.
Finally, the kernel blocks are defined for every code over $\Z_p$ containing
the zero vector, and the derived code for every block-homogeneous one.  It
would be interesting to use them for other families of GH codes obtained from
additive codes, such as the $\Z_p\Z_{p^2}$-linear GH codes
\cite{ZpZp2Construction,ZpZp2Classification}, and to decide which GH codes of
different families are equivalent.

{\small
\bibliographystyle{abbrv}
\bibliography{ref}
}
\end{document}